\documentclass{lipics-v2021}
\usepackage[utf8]{inputenc}

\nolinenumbers

\usepackage{comonads}
\usepackage{tikz-cd}
\usepackage{color}
\usepackage{todonotes}

\newcommand{\Scott}[3]{\varphi^{#1,#2}_{#3}}
\newcommand{\existc}[2]{\exists #1\;#2.\;}
\newcommand{\existx}[2]{\exists!#1\;#2.\;}
\newcommand{\acc}[2]{\alpha_{#1}(#2)}
\newcommand{\Acc}[2]{\mathsf{Acc}_{#1}(#2)}

\newtheorem{question}[theorem]{Question}

\title{Comonadic semantics for hybrid logic and bounded fragments}
\titlerunning{Comonadic semantics for hybrid logic and bounded fragments}
\author{Samson Abramsky}{UCL}{s.abramsky@ucl.ac.uk}{https://orcid.org/0000-0003-3921-6637}{}
\author{Dan Marsden}{University of Oxford}{daniel.marsden@cs.ox.ac.uk}{https://orcid.org/0000-0003-0579-0323}{}

\funding{EPSRC project EP/T00696X/1 ``Resources and Co-Resources: a junction between categorical semantics and descriptive complexity''}

\authorrunning{S. Abramsky and D. Marsden}

\Copyright{S. Abramsky and D. Marsden}

\begin{document}

\maketitle

\begin{abstract}
    In recent work, comonads and associated structures have been used to analyse a range of important notions in finite model theory, descriptive complexity and combinatorics. We extend this analysis to Hybrid logic, a widely-studied extension of basic modal logic, which corresponds to the bounded fragment of first-order logic. In addition to characterising the various resource-indexed equivalences induced by Hybrid logic and the bounded fragment, and the associated combinatorial decompositions of structures, we also give model-theoretic characterisations of bounded formulas in terms of invariance under generated substructures, in both the finite and infinite cases.
\end{abstract}

\section{Introduction}

Our starting point is recent work which exposes comonadic structures in finite model theory, descriptive complexity and combinatorics \cite{abramsky2017pebbling,abramsky2021relating}.
This work establishes a common template:
\begin{itemize}
\item A family of comonads $\{ \Ck \}$ on the category of relational structures, indexed by a resource parameter $k$.\footnote{In this paper, $k$ will always be a positive integer.}
\item These comonads induce resource-indexed equivalences on relational structures, which capture the equivalences induced by certain logical fragments.
Examples include the Ehrenfeucht-\Fraisse comonads $\Ek$, which capture the quantifier-rank fragments; the pebbling comonads $\Pk$, which capture the finite variable fragments; and the modal comonads $\Mk$, which capture the modal fragments of bounded modal depth.
\item The coalgebras $\alpha : \As \to \Ck \As$ for these comonads correspond to certain resource-bounded decompositions of a relational structure $\As$. The least $k$ for which such a coalgebra exists for $\As$ determines a significant combinatorial invariant of $\As$. In the case of $\Ek$, the invariant is the \emph{tree-depth} of the structure, while for $\Pk$, it is the \emph{tree-width} \cite{abramsky2021relating,abramsky2017pebbling}.
\end{itemize}
This entire framework is axiomatised in a very general setting in \cite{abramsky2021arboreal}.

In the present paper, we extend the program of comonadic semantics  to hybrid logic, and the corresponding bounded fragments of first-order logic. 
Hybrid logic (see e.g. \cite{blackburn2000representation,areces2001hybrid}) has been widely studied as an expressive extension of basic modal logic. It is semantically natural, e.g. in the analysis of temporal reasoning \cite{blackburn2016reichenbach}, and since it allows an internalisation of relational semantics, it has a very well-behaved proof theory \cite{brauner2010hybrid}, without needing to resort to explicit labelling of proofs or tableaux.
The corresponding fragment of first-order logic under modal translation is the \emph{bounded fragment}, in which quantification is relativized to atomic formulas from the relational vocabulary. This fragment is important in  set theory \cite{levy1965hierarchy}, and has been studied in general proof- and model-theoretic terms in \cite{feferman1966persistent,feferman1968persistent}.

The comonad which captures these logics is a natural restriction of a pointed version of the Ehrenfeucht-\Fraisse comonad previously introduced in \cite{abramsky2021relating}.
This comonadic analysis nicely reveals, in a clear and conceptual way, the way in which hybrid logic sits between basic modal logic and first-order logic.
We characterise the coalgebras for this comonad as tree covers of a relational structure with additional locality constraints.  This enables a uniform treatment of logical equivalences, bisimulation games, and combinatorial parameters, within the axiomatic framework recently given in \cite{abramsky2021arboreal}.
We also prove a model-theoretic characterization of invariance under formulas of these logics in both the finite and infinite cases, using an adaptation of Otto’s proof of the van Benthem-Rosen theorem for modal logic \cite{otto2004elementary}. As far as we know, the finite model theory version of this characterisation has not appeared previously.
We also extend this result to Hybrid temporal logic, which includes backwards modalities.

We shall treat the case of hybrid logic, which essentially amounts to restricting the relational vocabulary to unary predicates and a single binary relation, first and in some detail, since this is notationally simpler, and has been a focus of recent work. The 
extension to the general bounded fragment
is then outlined more briefly.

\section{Preliminaries}

We shall need a few notions on posets. Given $x,y \in P$ for a poset $(P, {\leq})$, we write $x \comp y$ if $x$ and $y$ are comparable in the order, \ie $x \leq y$ or $y \leq x$. We will use finite sequences extensively; these are partially ordered by prefix, with notation $s \preford t$.

A relational vocabulary $\sg$ is a set of relation symbols $R$, each with a specified positive integer arity.
A $\sg$-structure $\As$ is given by a set $A$, the universe of the structure, and for each $R$ in $\sg$ with arity $n$, a relation $\RA \subseteq A^n$. A homomorphism $h : \As \rarr \Bs$ is a function $h : A \rarr B$ such that, for each relation symbol $R$ of arity $n$ in $\sg$, for all $a_1, \ldots , a_n$ in $A$:
$\RA(a_1,\ldots , a_n) \IMP \RB(h(a_1), \ldots , h(a_n))$. We write $\CS$ for the category of $\sg$-structures and homomorphisms.

Since evaluation in modal logics is relative to a given world, we shall also use the pointed category $\CSp$. Objects are pairs $(\As, a)$, where $\As$ is a $\sg$-structure, and $a \in A$. Morphisms $h : (\As, a) \to (\Bs, b)$ are homomorphisms $h : \As \to \Bs$ such that $h(a) = b$.

A \emph{modal vocabulary} has only relation symbols of arity $\leq 2$: a set of unary predicate symbols $P$, which will correspond to modal propositional atoms; and a family $\{ E_i \}$ of binary relations, which we think of as \emph{transition relations}.\footnote{More traditionally referred to as \emph{accessibility relations}.}
We will say that ``$a$ sees $b$'' if $\EAi(a,b)$ for some $i$. The \emph{unimodal} case is where there is a single transition relation $E$.

\subsection{Hybrid logic}
\label{hlsec}

Hybrid logic formulas are built from propositional atoms $p$ and \emph{world variables} $x$, with the following syntax:
\[ \vphi \;\; ::= \;\; p \mid x \mid \neg \vphi \mid \vphi \wedge \vphi' \mid  \vphi \vee \vphi'  \mid \Box \vphi \mid \Diamond \vphi \mid \da x. \, \vphi \mid @_{x} \vphi . \]
We use a redundant syntax to make it more convenient to discuss fragments.
The new features compared with standard basic modal logic are the world variables, which can be bound with $\da$, and used to force evaluation at a given world with $@$.
Hybrid formulae are graded by their \emph{hybrid modal depth}. This is the usual notion of modal depth, with the adjustment that sub-formulae of the form~$\Diamond x$
, for some world variable $x$, are deemed to have zero depth.

The semantics of hybrid logic is given by translation into first-order logic with equality over a unimodal vocabulary, with a unary predicate  $P$ for each proposition atom $p$, and a single transition  relation $E$.
World variables are treated as ordinary first-order variables. The translation is parameterised on a variable, corresponding to the world at which the formula is to be evaluated.
We write $\psi[x/y]$ for the result of substituting $x$ for the free occurrences of $y$ in $\psi$.
\[ \begin{array}{lcl}
\ST_{x}(p) & = & P(x) \\
\ST_{x}(x') & = & x = x' \\
\ST_{x}(\neg \vphi) & = & \neg \, \ST_{x}(\vphi) \\
\ST_{x}(\vphi \wedge \vphi') & = & \ST_{x}(\vphi) \wedge \ST_{x}(\vphi') \\
\ST_{x}(\vphi \vee \vphi') & = & \ST_{x}(\vphi) \vee \ST_{x}(\vphi') \\
\ST_{x}(\Box \vphi) & = & \forall y. [ E(x,y) \to \ST_{y}(\vphi)] \\
\ST_{x}(\Diamond \vphi) & = & \exists y. [ E(x,y) \wedge \ST_{y}(\vphi)] \\
\ST_{x}(\da x'. \vphi) & = & \ST_{x}(\vphi)[x/x'] \\
\ST_{x}(@_{x'} \vphi) & = & \ST_{x}(\vphi)[x'/x]
\end{array}
\]

The target of this translation is the \emph{bounded fragment} of first-order logic with equality, with quantifiers restricted to those of the form $\exists y. [ E(x,y) \wedge \vphi]$, 
$\forall y. [ E(x,y) \to \vphi]$, with $x \neq y$.
Hybrid logic is in fact equiexpressive with this fragment \cite{areces2001hybrid}. We shall discuss the bounded fragment  in section~\ref{bdfragsec}. 

Note that $\ST_{x}(\Diamond y)$ is logically equivalent to $E(x,y)$. Thus this formula tests for the presence of a transition between worlds which have already been reached, justifying our assignment of modal depth $0$.

One feature of hybrid logic which we have omitted here is \emph{nominals}, which correspond to constants under the first-order translation. These will also be included in our treatment of the bounded fragment in section~\ref{bdfragsec}.

\section{The hybrid comonad}

We shall now introduce the hybrid comonad on $\CSp$ for modal vocabularies $\sg$, motivating it as combining features of the Ehrenfeucht-\Fraisse and modal comonads from \cite{abramsky2021relating}.
\begin{itemize}
    \item 
We recall firstly the Ehrenfeucht-\Fraisse comonad $\Ek$ on $\CS$ for an arbitrary vocabulary $\sg$. Given a structure $\As$, the universe of $\Ek \As$ is the set of non-empty sequences of elements of $A$ of length $\leq k$. We think of these sequences as \emph{plays} in the Ehrenfeucht-\Fraisse game on $\As$. We define the map $\epsA : \Ek A \to A$ which sends a sequence to its last element, which we think of as the current move or \emph{focus} of the play. For a relation $R$ of arity $n$, we define $R^{\Ek \As}(s_1,\ldots,s_n)$ to hold iff $s_i \comp s_j$ for all $1\leq i,j \leq n$, and $\RA(\epsA(s_1), \ldots , \epsA(s_n))$. Explicitly, for unary predicates $P$, $P^{\Ek \As}(s)$  iff $P^{\As}(\epsA(s))$, and for a binary  relation $R$, $R^{\Ek \As}(s,t)$ iff $s \comp t$ and  $\RA(\epsA(s), \epsA(t))$. Thus the relations hold \emph{along} plays as one extends another, but not between \emph{different} (\ie incomparable) plays.

\item This construction lifts to the pointed category $\CSp$. We define the universe of $\Ek (\As,a)$ to comprise the non-empty sequences of length $\leq k+1$ which start with $a$. The distinguished element is $\langle a \rangle$. The relations are lifted in exactly the same way as previously.

\item The modal comonad $\Mk$ over a unimodal vocabulary with unary predicates $P$ corresponding to propositional atoms, and a single transition relation $E$, restricts the sequences in $\Ek (\As,a)$ to those of the form $\langle a_0,\ldots a_j\rangle$, $a_0 = a$, such that for all $i$ with $0 \leq i < j$, $\EA(a_i,a_{i+1})$. Thus we can only extend a sequence with an element which the previous element ``sees''. Moreover, the transition relation $E$ is lifted in a correspondingly local fashion, so that a sequence is only related to its immediate extensions: $E^{\Mk (\As,a)}(s,t)$ iff  $t = s\langle a \rangle$ and $\EA(\epsA(s), \epsA(t))$.
This is the familiar \emph{unravelling} construction for modal structures \cite{blackburn2002modal}.

\item The hybrid comonad $\Hk$ is again defined on the pointed category $\CSp$ over a unimodal vocabulary. $\Hk (\As,a)$ has as universe the subset of $\Ek (A,a)$ of those sequences $\langle a_0, a_1, \ldots , a_{l}\rangle$ such that $a_0 =a$, and for all $j$ with $0 < j \leq l$, for some $i$, $0 \leq i < j$, $\EA(a_{i},a_{j})$. Thus we relax the locality condition of $\Mk$ to the condition that a sequence can only be extended with an element if it is seen by \emph{some} element which has been played previously. The $\sg$-relations on $\Hk (\As,a)$ are defined exactly as for $\Ek (\As,a)$, and the distinguished element is $\langle a \rangle$, so  $\Hk (\As,a)$ is the induced substructure of $\Ek (\As,a)$ given by this restriction of the universe. In this sense, $\Hk$ is closer to $\Ek$ than to $\Mk$.
\end{itemize}

To complete the specification of $\Hk$, we define the \emph{coKleisli extension}: given a morphism $h : \Hk (\As,a) \to (\Bs,b)$, we define $h^* : \Hk (\As,a) \to \Hk (\Bs,b)$ by 
\[ h^*(\langle a, a_1, \ldots , a_i \rangle) = \langle h(\langle a \rangle), h(\langle a, a_1 \rangle), \ldots , h(\langle a, a_1, \dots , a_i \rangle) \rangle . \]
We can verify that for each structure $\As$, $\epsA : \Hk \As \to \As$ is a morphism; that for each morphism $h : \Hk (\As,a) \to (\Bs,b)$, $h^* : \Hk (\As,a) \to \Hk (\Bs,b)$ is a morphism; and that the following equations are satisfied, for all morphisms $h : \Hk (\As,a) \to (\Bs,b)$, $g : \Hk (\Bs,b) \to (\Cs,c)$:
\[ \epsA \circ h^* = h, \qquad \epsA^* = \id_{\Hk \As}, \qquad (g \circ h^*)^* = g^* \circ h^* , \]
This establishes the following result.
\begin{proposition}
\label{Hkprop}
The triple $(\Hk, \ve, (\cdot)^{*})$ is a comonad in Kleisli form \cite{manes2012algebraic}.
\end{proposition}
\begin{proof}
The proof of the above items largely carries over from the corresponding arguments for $\Ek$ \cite[Proposition 3.1]{abramsky2021relating}. The additional point to be checked is that $h^*$ cuts down to the substructures induced by the hybrid comonad. This follows immediately from the definition of  $E^{\Hk \As}(s,t)$, and the fact that $h$ is a homomorphism.
\end{proof}
It is then standard  \cite{manes2012algebraic} that $\Hk$ extends to a functor by $\Hk f = (f \circ \epsilon)^*$; that $\ve$ is a natural transformation; and that if we define the comultiplication $\delta : \Hk \Rightarrow \Hk^2$ by $\delta_{\As} = \id_{\Hk \As}^*$, then $(\Hk, \ve, \delta)$ is a comonad.

\subsection{\texorpdfstring{$I$-morphisms}{I-morphisms} and equality}
\label{Imorsec}

Like the Ehrenfeucht-\Fraisse comonad $\Ek$, and unlike the modal comonad $\Mk$, equality is important for $\Hk$, as we might expect from its appearance in the translation of hybrid logic into first-order logic. We shall follow the procedure introduced in  \cite[Section 4]{abramsky2021relating} to ensure that equality is properly handled in $\Ek$.

The issue is that elements of $A$ may be repeated in the plays in $\Hk (\As,a)$. In particular, this happens when there are cycles in the graph $(A, \EA)$ which are reachable from $a$. We wish to view coKleisli morphisms $f : \Hk (\As,a) \to (\Bs,b)$ as winning strategies for Duplicator in the one-sided (or existential) Spoiler-Duplicator game from $(\As,a)$ to $(\Bs,b)$, in which Spoiler plays in $\As$ and Duplicator in $\Bs$ \cite{kolaitis1990expressive}. In order to fulfil the partial homomorphism winning condition, $f$ must map repeated occurrences of an element $a' \in A$ in a play $s$ in $\Hk (\As,a)$ to the same element of $B$. 
The same issue will recur when we deal with back-and-forth games in section~\ref{opensec}.
We seek a systematic means of enforcing this requirement.

Given a relational vocabulary $\sg$, we produce a new one $\sgp  = \sigma \cup \{I\}$, where $I$ is a binary relation symbol not in $\sg$. 
If we interpret $I^{(\As,a)}$ and $I^{(\Bs,b)}$ as the identity relations on $A$ and $B$, then, following the general prescription for relation lifting in $\Ek (\As,a)$, and hence also in $\Hk (\As,a)$ as an induced substructure of $\Ek (\As,a)$, we have $I^{(\Hk \As,a)}(s,t)$ iff $s \comp t$ and $\epsA(s) = \epsA(t)$.
Thus a $\sg$-morphism $f : (\Hk \As,a) \to (\Bs,b)$ satisfies the required condition iff it is a $\sgp$-morphism.

As it stands, this is an ad hoc condition: it relies on a special interpretation of the $I$-relation. We want our objects to live in $\CSp$, but our morphisms to live in $\CSpp$.
To accomplish this, we use a simple special case of the notion of \emph{relative comonad} \cite{altenkirch2010monads}. 
We can take advantage of the fact that $\Ek$, and hence $\Hk$ as a sub-comonad of $\Ek$, is defined uniformly in the vocabulary. 
Given a vocabulary $\sg$, there is a full and faithful embedding $J : \CSp \to \CSpp$ such that $I^{J (\As,a)}$ is the identity on $A$.
Moreover, we have a comonad $\EkI$, which is the $\Ek$ construction applied to $\CSpp$. Note that this treats $I$ like any other binary relation in the vocabulary.

We correspondingly obtain $\HkI (\As, a)$ as the substructure of $\EkI (\As,a)$ induced by restricting the universe to that of $\Hk(\As,a)$. It is important to note that only the transition relation $E$ is used to restrict the universe. 

We use this data to obtain the $J$-relative comonad $\Hkplus = \HkI \circ J$ on $\CSp$. The objects of the coKleisli category for this relative comonad are those of $\CSp$.
CoKleisli morphisms have the form $\HkI J (\As,a) \to J (\Bs,b)$. The counit and coextension are the restrictions of those for $\HkI$ to the image of $J$.

A more general version of this construction will be given in our discussion of the bounded fragment in section~\ref{bdfragsec}.

\subsection{CoKleisli maps, existential games, and the existential positive fragment}

The standard  $k$-round existential Ehrenfeucht-\Fraisse game from $\As$ to $\Bs$ \cite{kolaitis1990expressive,abramsky2021relating} is defined as follows. In each round $i$, Spoiler moves by choosing an element $a_i$ from $A$, and Duplicator responds by choosing an element $b_i$ from $B$. The winning condition for Duplicator is that the correspondence $a_i \mapsto b_i$ is a partial homomorphism from $\As$ to $\Bs$.

The $k$-round existential hybrid game from $(\As,a)$ to $(\Bs,b)$ is defined in exactly the same way, with two additional provisos:
\begin{itemize}
\item At round $0$, Spoiler must play $a_0 = a$, and Duplicator must respond with $b_0 = b$.
\item At round $j>0$, Spoiler must play a move $a_{j}$ such that, for some $i<j$, $\EA(a_i,a_j)$.
\end{itemize}

\begin{proposition}
There is a bijective correspondence between
\begin{itemize}
\item Winning strategies for Duplicator in the $k$-round existential hybrid game from $(\As,a)$ to $(\Bs,b)$
\item CoKleisli morphisms $h : \Hkplus (\As,a) \to J (\Bs,b)$.
\end{itemize}
\end{proposition}
\begin{proof}
This is mostly the same argument as in the proofs of Theorems~3.2 and~5.1 in \cite{abramsky2021relating}. The additional point to note is that the conditions on Spoiler's moves match those  used in defining the universe of $\Hkplus (\As,a)$ as a subset of that of $\Ek (\As,a)$. Moreover, the partial homomorphism winning condition ensures that the same conditions are respected by the responses of Duplicator.
\end{proof}
The existential positive fragment $\HLp$ of hybrid logic is defined by omitting negation and $\Box$ in the syntax for hybrid logic given in section~\ref{hlsec}.
$\HLp_{k}$ is the fragment of $\HLp$ comprising formulas of hybrid modal depth~$\leq k$.

This fragment induces a preorder on pointed structures: 
\[ (\As,a) \SPreord_k (\Bs,b) \; \stackrel{\Delta}{\Longleftrightarrow} \; \forall \vphi \in \HLp_k. \, [\Asa \models \vphi \IMP \Bsb \models \vphi]. \]
Here by $\Asa \models \vphi$ we mean $\Asa \models \psi(x)$, where $\psi(x) = \ST_{x}(\vphi)$.

We define another preorder on pointed structures: $(\As,a) \rightarrow^{\mathbb{H}}_{k} (\Bs,b)$ iff there is a coKleisli morphism  $h : \Hkplus (\As,a) \to J (\Bs,b)$.

\begin{theorem}
Let $\sg$ be a finite unimodal vocabulary.
For all  $(\As,a)$, $(\Bs,b)$ in $\CSp$:
\[  (\As,a) \SPreord_k (\Bs,b) \; \IFF \; (\As,a) \to^{\mathbb{H}}_{k} (\Bs,b) . \]
\end{theorem}

\section{Coalgebras}
\label{sec:coalgebras}
We now study coalgebras for the hybrid comonad.
A coalgebra for a comonad $(G, \varepsilon, \delta)$ is a morphism $\alpha : A \to G A$ such that the following diagrams commute:
\begin{center}
\begin{tikzcd}
A \ar[r, "\alpha"] \ar[rd, "\id_A"']
& G A \ar[d, "\ve_{A}"] \\
& A
\end{tikzcd}
$\qquad \qquad$
\begin{tikzcd}
A  \ar[r, "\alpha"] \ar[d, "\alpha"']
& G A \ar[d,  "\delta_{A}"] \\
G A  \ar[r, "G \alpha"] 
& G^2 A
\end{tikzcd}  
\end{center}
Given $G$-coalgebras $\alpha : A \to GA$ and $\beta : B \to GB$, a coalgebra morphism from $\alpha$ to $\beta$ is a morphism $h : A \to B$ such that the following diagram commutes:
\begin{center}
\begin{tikzcd}
A \ar[r, "\alpha"] \ar[d, "h"']
& GA \ar[d, "Gh"] \\
B \ar[r, "\beta"'] & GB
\end{tikzcd}
\end{center}
This gives a category of coalgebras and coalgebra morphisms, denoted by $\EM(G)$, the \textit{Eilenberg-Moore category} of $G$.

We will now analyze $\EMHk$, the category of coalgebras for the hybrid comonad on a unimodal vocabulary $\sg$. This will lead to  a natural combinatorial parameter associated with hybrid logic and the bounded fragment, which is a refinement of \emph{tree-depth} \cite{nevsetvril2006tree}. It will also provide a basis for a comonadic characterisation of bisimulation and the equivalence on structures induced by the full hybrid logic, as we will see in the next section.

We will need a few more notions on posets.
A chain in a poset $(P, {\leq})$ is  a subset $C \subseteq P$ such that, for all $x, y \in C$, $x \comp y$. A \emph{forest} is a poset $(F, {\leq})$ such that, for all $x \in F$, the  set of predecessors $\da(x) \, := \, \{ y \in F \mid y \leq x\}$ is a finite chain. The height $\hgt(F)$ of a forest $F$ is $\sup_{C} | C |$, where $C$ ranges over chains in $F$.
Note that the height is either finite or $\omega$.
A \emph{tree} is a forest with a least element  (the root).
We write the covering relation for a poset as $\cvr$; thus $x \cvr y$ iff $x \leq y$, $x \neq y$, and for all $z$, $x \leq z \leq y$ implies $z = x$ or $z = y$.
Morphisms of trees are monotone maps preserving the root and the covering relation.

Given a $\sg$-structure $\As$, the \emph{Gaifman graph} $\Gf(\As)$ is $(A, \adj)$, where $a \adj a'$ ($a$ is adjacent to $a'$) if they are distinct elements of $A$ which both occur in a tuple of some relation $\RA$, $R$ in $\sg$.

A \emph{tree cover} of a pointed $\sg$-structure $(\As,a)$ is a tree order $(A, {\leq})$ on $A$ with least element $a$, and such that if $a \adj a'$, then $a \comp a'$. Thus adjacent elements in the Gaifman graph must appear in the same branch of the tree. The tree cover is \emph{generated} if for all $a' \in A$ with $a' \neq a$, for some $a'' \in A$, $a'' < a'$ and $\EA(a'',a')$.

\begin{theorem}
For any pointed $\sg$-structure $(\As,a)$, and $k > 0$, there is a bijective correspondence between:
\begin{itemize}
\item $\Hk$-coalgebras $\alpha : (\As,a) \to \Hk (\As,a)$.
\item Generated tree covers of $(\As,a)$ of height $\leq k + 1$.
\end{itemize}
\end{theorem}

We define the \emph{generated tree depth} of $(\As,a)$ to be the minimum height of any generated tree cover of $(\As,a)$. This can be seen as a refinement of the standard notion of tree depth  \cite{nevsetvril2006tree}.

We define the \emph{hybrid coalgebra number} of $(\As,a)$ to be the least $k$ such that there is an $\Hk$-coalgebra $\alpha : (\As,a) \to \Hk (\As,a)$. If there is no coalgebra for any $k$, the hybrid coalgebra number is $\omega$.
\begin{theorem}
The generated tree depth of a structure $(\As,a)$ coincides with its hybrid coalgebra number.
\end{theorem}

We define a category $\TS$ with objects $(\As,a,{\leq})$, where $(\As,a)$ is a pointed $\sg$-structure, and $\leq$ is a generated tree cover of $(\As,a)$. Morphisms $h : (\As,a,{\leq}) \to (\Bs,b,{\leq'})$ are morphisms of pointed $\sg$-structures which are also tree morphisms. For each $k>0$, there is a full subcategory $\TSk$ determined by those objects whose covers have height $\leq k$.
\begin{theorem}
\label{TSkthm}
For each $k>0$, $\TSk$ is isomorphic to $\EM(\Hk)$.
\end{theorem}

There is an evident forgetful functor $U_k : \TSk \to \CSp$ which sends $(\As, a, {\leq})$ to $(\As,a)$.
\begin{theorem}
For each $k>0$, $U_k$ has a right adjoint $R_k : \CSp \to \TSk$ given by $R_k (\As,a) = (\Hk (\As,a), {\preford})$.
The comonad induced by this adjunction is $\Hk$. The adjunction is comonadic.
\end{theorem}

\section{Paths, open maps, and back-and-forth equivalence}
\label{opensec}

The coalgebra category $\EM(\Hk)$ has a richer structure than $\CSp$, articulated as $\TSk$ by Theorem~\ref{TSkthm}.
In fact, $\TSk$ is an \emph{arboreal category} as defined in \cite{abramsky2021arboreal}.
This allows us to define notions of bisimulation and games on this category, which can then be transferred to $\CSp$ via the adjunction $U_k \dashv R_k$, following the general pattern laid out in \cite{abramsky2021relating}.

To accommodate $I$-morphisms, as discussed in section~\ref{Imorsec}, we work with the $J$-relative version of this adjunction, using $\Rkp = \RkI J$, where $\RkI$ is the instance of the adjunction for $\CSpp$.

We now follow the same script as detailed in \cite{abramsky2021relating}, and axiomatised  in \cite{abramsky2021arboreal}.

\subsection{Embeddings, paths and pathwise embeddings}
A morphism $e$ in $\TSk$ is an \emph{embedding} if
$U_k(e)$ is an embedding of relational structures. 
We write  $e:  T \rightarrowtail U$ to indicate that $e$ is an embedding.

A \emph{path} in $\TSk$ is an object $P$ such that the associated tree cover is a finite linear order,
so it comprises a single branch. If $P$ is a path, then $I^{P}$ is the identity relation. We say that $e : P \embed T$ is a \emph{path embedding} if $P$ is a path. A morphism $f : T \to U$ in $\TSk$ is a \emph{pathwise embedding} if for any path embedding $e : P \embed T$, $f \circ e$ is a path embedding.

\subsection{Open maps}

A morphism $f : T \to U$ in $\TSk$ is \emph{open} if, whenever we have a diagram
\[ \begin{tikzcd}
  P \arrow[r, rightarrowtail] \arrow[d,rightarrowtail]
    & Q \arrow[d, rightarrowtail] \\
  T \arrow[r,  "f"']
& U
\end{tikzcd}
\]
where $P$ and $Q$ are paths, there is an embedding $Q \rightarrowtail T$ such that
\[ \begin{tikzcd}
  P \arrow[r, rightarrowtail] \arrow[d,rightarrowtail]
    & Q \arrow[dl, rightarrowtail] \arrow[d, rightarrowtail] \\
  T \arrow[r,  "f"']
& U
\end{tikzcd}
\]
This is often referred to as the \emph{path-lifting property}. If we think of $f$ as witnessing a simulation of $T$ by $U$, path-lifting means that if we extend a given behaviour in $U$ (expressed by extending the path $P$ to $Q$), then we can find a matching behaviour in $T$ to ``cover'' this extension. Thus it expresses an abstract form of the notion of  ``p-morphism'' from modal logic \cite{blackburn2002modal}, or of functional bisimulation.

\subsection{Bisimulation}

We can now define the \emph{back-and-forth equivalence} $(\As,a) \eqbHk (\Bs,b)$ between structures in $\CSp$. This holds if there is a span of open pathwise embeddings in $\TSk$ 
\[ \begin{tikzcd}
& T \arrow[dl] \arrow[dr] \\
\Rkp (\As,a) & & \Rkp (\Bs, b)
\end{tikzcd}
\]
Note that we are using the arboreal category $\TSk$ to define an equivalence on the ``extensional category'' $\CSp$.

\subsection{Games}

We shall now define a \emph{back-and-forth game} $\Gk((\As,a),(\Bs,b))$ played between  $(\As,a)$ and $(\Bs,b)$, using  the comonad $\Hk$. Positions of the game are pairs $(s,t) \in \Hk (\As,a) \times \Hk (\Bs,b)$. The initial position is $(\langle a \rangle, \langle b \rangle)$. 

We define a relation  $\WABCp$ on positions  as follows. A pair $(s,t)$ is in $\WABCp$ iff for some path $P$, path embeddings $e_1 : P \embed \Hk (\As,a)$ and $e_2 : P \embed \Hk (\Bs,b)$, and $p \in P$, $s = e_1(p)$ and $t = e_2(p)$. The intention is that $\WABCp$ picks out the winning positions for Duplicator.

At the start of each round of the game, the position is specified by $(s, t) \in \Hk (\As,a) \times \Hk (\Bs,b)$.  The round proceeds as follows. Either Spoiler chooses some $s' \rcvr s$, and Duplicator must respond with $t' \rcvr t$, resulting in a new position $(s', t')$; or Spoiler chooses some $t'' \rcvr t$ and Duplicator must respond with $s'' \rcvr s$, resulting in $(s'',t'')$. Duplicator wins the round if they are able to respond, and  the new position is in $\WABCp$.

\subsection{Results}

\begin{theorem}
Given $(\As,a)$, $(\Bs,b)$ in $\CSp$, the following are equivalent:
\begin{enumerate}
\item $(\As,a) \eqbHk (\Bs,b)$
\item Duplicator has a winning strategy for $\Gk((\As,a),(\Bs,b))$.
\end{enumerate}
\end{theorem}
\begin{proof}
The proof is a minor variation of that for \cite[Theorem 10.1]{abramsky2021relating}, the corresponding result for $\Ek$. Alternatively, this is an instance of  the very general \cite[Theorem 6.9]{abramsky2021arboreal}.
\end{proof}

The standard  $k$-round Ehrenfeucht-\Fraisse game between $\As$ and $\Bs$ \cite{Libkin2004} is defined as follows. In each round $i$, Spoiler moves by either
\begin{itemize}
\item choosing an element $a_i$ from $A$, to which Duplicator responds by choosing an element $b_i$ from $B$; or
\item choosing an element $b_i$ from $B$, to which Duplicator responds by choosing an element $a_i$ from $A$.
\end{itemize}
The winning condition for Duplicator is that the correspondence $a_i \mapsto b_i$ is a partial isomorphism from $\As$ to $\Bs$.

The $k$-round back-and-forth hybrid game between  $(\As,a)$ and $(\Bs,b)$ is defined in exactly the same way, with two additional provisos:
\begin{itemize}
\item At round $0$, Spoiler must either play $a_0 = a$, to which Duplicator must respond with $b_0 = b$; or $b_0 = b$, to which Duplicator must respond with $a_0 = a$
\item At round $j>0$, if Spoiler plays a move $a_{j} \in A$ then, for some $i<j$, $\EA(a_i,a_j)$; while if Spoiler plays a move $b_{j} \in B$ then, for some $i<j$, $\EB(b_i,b_j)$.
\end{itemize}
The partial isomorphism winning condition ensures that Duplicator is subject to the same constraints.

We write $\HLk$ for the set of hybrid formulas of modal depth $k$.
We define an equivalence relation on pointed structures by:
\[ (\As,a) \Hequivk (\Bs,b) \; \stackrel{\Delta}{\Longleftrightarrow} \; \forall \vphi \in \HLk. \, [\Asa \models \vphi \IFF \Bsb \models \vphi]. \]
\begin{theorem}
\label{Hequivth}
Let $\sg$ be a finite unimodal vocabulary.
For all  $(\As,a)$, $(\Bs,b)$ in $\CSp$, the following are equivalent:
\begin{enumerate}
\item $(\As,a) \eqbHk (\Bs,b)$.
\item Duplicator has a winning strategy for the $k$-round back-and-forth hybrid game between  $(\As,a)$ and $(\Bs,b)$.
\item $(\As,a) \Hequivk (\Bs,b)$.
\end{enumerate}
\end{theorem}

\section{The general bounded case}
\label{bdfragsec}

We now turn to a treatment of the bounded comonad $\Bk$, corresponding to the bounded fragment of first-order logic. This generalises the hybrid comonad in two directions:
\begin{itemize}
\item We allow for constants in the vocabulary.
\item We also allow arbitrary relational vocabularies.
\end{itemize}
Allowing for constants $c_1,\ldots,c_m$ involves working with the $m$-pointed category $\CSm$. This has objects $(\As, \va)$, where $\va = \langle a_1, \ldots , a_m \rangle \in A^m$.
Morphisms $h : (\As,\va) \to (\Bs, \vb)$ must preserve these tuples. The intention is that $a_i = c_i^{\As}$. Note that $\CSp = \mathsf{Struct}_{1}(\sg)$.

The comonad $\Ek$ extends straightforwardly to $\CSm$. The idea is that the Ehrenfeucht-\Fraisse game on $(\As,\va)$ is now played with the proviso that the first $m$ elements are predetermined to be the tuple $\va$. Thus the universe of $\Ek (\As,\va)$ comprises non-empty sequences $s$ of length $\leq k+m$, such that $s \comp \va$. The $\sg$-relations are lifted to $\Ek (\As,\va)$,  and the counit $\epsA$ and coKleisli extension $f^*$ are defined, in the exactly the same way as previously. It is straightforward to see that this gives a well-defined comonad on $\CSm$, for any relational vocabulary $\sg$.

The generalization to arbitrary vocabularies involves a slight subtlety, since we wish to control which relations can be used to bound quantifiers. We define a \emph{bounded vocabulary} to be a pair $(\sg,\sgT)$, where $\sg$ is an arbitrary relational vocabulary, and $\sgT \subseteq \sg$ is a sub-vocabulary consisting only of binary relations, which we refer to as \emph{transition relations}. Note that a unimodal vocabulary is the special case where $\sgT = \{ E \}$, and $\sg = \sgT \cup \sg_{1}$, where $\sg_{1}$ is a set of unary predicates.

The category $\CSmT$ is defined in exactly the same fashion as $\CSm$. This may seem to render the notion of bounded vocabulary redundant, but in fact there is some indexed structure at play: we have an indexed family of categories, and an indexed family of comonads defined on them, and the comonad constructions can make use of the additional information given by the specification of $\sgT$.

The $\Ek$ comonad is defined on $\CSmT$ in exactly the same fashion as on $\CSm$. We define the $\Bk$ comonad as follows. The universe of $\Bk (\As,\va)$ is obtained by restricting the universe of $\Ek (\As,\va)$ to those sequences $s = \langle a_1, \ldots , a_l \rangle$ such that, for all $j: m < j \leq l$, for some $i: 1 \leq i < j$, $\EA(a_i,a_j)$ for some transition relation $E$. We then define $\Bk (\As,\va)$ as the induced substructure of $\Ek (\As,\va)$ determined by this restriction of the universe. The definition of counit and coKleisli extension carry over from $\Ek$, in the same way as for the hybrid comonad $\Hk$. Thus we obtain the analogue of Proposition~\ref{Hkprop}:
\begin{proposition}
The triple $(\Bk, \ve, (\cdot)^{*})$ is a comonad in Kleisli form.
\end{proposition}

Bounded vocabularies allow a more systematic treatment of $I$-morphisms. Given a bounded vocabulary $(\sg,\sgT)$, we define $(\sg, \sgT)^{+} = (\sgp,\sgT)$. This makes explicit that the $I$ relation is not a transition relation. This operation also takes bounded vocabularies to bounded vocabularies, whereas it is not closed on unimodal vocabularies.
Thus we can define $\BkI$ uniformly as the bounded comonad defined on $\CSmT^{+}$; and the $J$-relative bounded comonad $\Bkp = \BkI J$.

Turning to the connection with logic, we are concerned with the bounded fragment of first-order logic over a bounded vocabulary $(\sg,\sgT)$, and constants $c_1,\ldots,c_m$.
Terms are constants or variables. Formulas are formed according to the following syntax:
\[ \psi \;\; ::=\;\; R(t_1, \ldots , t_n) \mid t = u \mid \neg \psi \mid \psi \wedge \psi' \mid  \psi \vee \psi' \mid \forall x. [E(t,x) \to \psi] \mid \exists x. [E(t,x) \wedge \psi]. \]
Here $R$ is a relation of arity $n$ in $\sg$, and $E$ is a transition relation in $\sgT$. The last two clauses are subject to the stipulation that $x \neq t$.

Classical contexts for bounded quantification include set theory \cite{barwise2017admissible}: $\forall x. (x \in y) \to \psi$, and bounded arithmetic \cite{buss1985bounded}: $\forall x. (x < y) \to \psi$. For a discussion of the bounded fragment, and comparison with guarded fragments, see \cite{van2005guards}.
A systematic study of properties of this fragment was made in \cite{feferman1966persistent,feferman1968persistent}. It is worth noting that in the formalism used in \cite{feferman1968persistent}, there is, implicitly, a designated ``transition relation''.

We can extend our translation of hybrid logic into the bounded fragment from section~\ref{hlsec} to include \emph{nominals} \cite{blackburn2000representation,areces2001hybrid}. We identify the nominals syntactically with the first-order constants $c_i$. The syntax of hybrid logic is extended with nominals as atoms, and as arguments to the evaluation modality. The translation is extended with the following clauses:
\[ \begin{array}{lcl}
\ST_{x}(c_i) & = & x = c_i \\
\ST_{x}(@_{c_i} \vphi) & = & \ST_{x}(\vphi)[c_i/x].
\end{array}
\]
We write $\BF$ for the sentences of the bounded fragment, and $\BFp$ for the existential positive part, obtained by omitting negation and universal quantification from the formation rules. We write $\BFk$, $\BFpk$  for the restriction of these fragments to sentences of quantifier rank $\leq k$. 

The fragment $\BFpk$ induces a preorder on objects of $\CSmT$:
\[ (\As,\va) \BFPreord_k (\Bs,\vb) \; \stackrel{\Delta}{\Longleftrightarrow} \; \forall \psi \in \BFpk. \, [\Ava \models \psi \IMP \Bvb \models \psi]. \]
To characterise this preorder, we generalise the existential hybrid game to the bounded setting.
The $k$-round existential bounded game from $(\As,\va)$ to $(\Bs,\vb)$ is played in the same fashion as the $k+m$-round existential Ehrenfeucht-\Fraisse game, with two additional provisos:
\begin{itemize}
\item At round $i : 1 \leq i \leq m$, Spoiler must play $a_i$ from the tuple $\va$, and Duplicator must respond with $b_i$ from $\vb$.
\item At round $j>m$, Spoiler must play a move $a_{j}$ such that, for some $1 \leq i <j$, $\EA(a_i,a_j)$ for some transition relation $E$.
\end{itemize}
\begin{proposition}
Let $(\sg, \sgT)$ be a finite bounded vocabulary.
For all  $(\As,\va)$, $(\Bs,\vb)$ in $\CSmT$, the following are equivalent:
\begin{itemize}
\item $(\As,\va) \BFPreord_k (\Bs,\vb)$
\item Duplicator has a winning strategy for the $k$-round existential bounded game from $(\As,\va)$ to $(\Bs,\vb)$.
\end{itemize}
\end{proposition}
We can now characterise the preorder in terms of the bounded comonad.
We define $(\As,\va) \rightarrow^{\mathbb{B}}_{k} (\Bs,\vb)$ iff there is a coKleisli morphism  $h : \Bkp (\As,\va) \to J (\Bs,\vb)$.
\begin{theorem}
Let $(\sg, \sgT)$ be a finite bounded vocabulary.
For all  $(\As,\va)$, $(\Bs,\vb)$ in $\CSmT$:
\[  (\As,\va) \BFPreord_k (\Bs,\vb) \; \IFF \;  (\As,\va) \rightarrow^{\mathbb{B}}_{k} (\Bs,\vb) . \]
\end{theorem}
To extend this result to the full bounded fragment, once again we study the coalgebras of the comonad. 

A \emph{tree cover} of $(\As,\va)$ is a tree order $(A, {\leq})$ such that the restriction of the order to $\{ a_1,\ldots, a_m\}$ is the chain $a_1 \cvr \cdots \cvr a_m$, and the restriction to $A \setminus \{ a_1, \ldots , a_{m-1} \}$ is a tree with root $a_m$. Moreover, if $a \adj a'$, then $a \comp a'$.
The tree cover is \emph{generated} if for all $a \in A$ with $a > a_m$, for some $a' \in A$, $a' \leq a$ and $\EA(a',a)$ for some transition relation $E$.
The height of the tree cover is $\hgt(A, {\leq}) - m$.

We define a category $\TSb$ with objects $(\As,\va,{\leq})$, where $(\As,\va)$ is an $m$-pointed $\sg$-structure, and $\leq$ is a generated tree cover of $(\As,\va)$. Morphisms $h : (\As,\va,{\leq}) \to (\Bs,\vb,{\leq'})$ are morphisms of $m$-pointed $\sg$-structures which are also tree morphisms. For each $k>0$, there is a full subcategory $\TSbk$ determined by those objects whose covers have height $\leq k$.
\begin{theorem}
For each $k>0$, $\TSbk$ is isomorphic to $\EM(\Bk)$.
\end{theorem}

There is an evident forgetful functor $U_k : \TSbk \to \CSmT$ which sends $(\As, \va, {\leq})$ to $(\As,\va)$.
\begin{theorem}
For each $k>0$, $U_k$ has a right adjoint $R_k : \CSmT \to \TSbk$ given by $R_k (\As,\va) = (\Bk (\As,\va), {\preford})$.
The comonad induced by this adjunction is $\Bk$. The adjunction is comonadic.
\end{theorem}

We can now follow the same script as in sections~5.1--5.4 to define open pathwise embeddings in $\TSbk$, and the back-and-forth equivalence $(\As,\va) \eqbBk (\Bs,\vb)$ between structures in $\CSmT$.

To connect back to the bounded fragment, we define the equivalence it induces on $m$-pointed structures:
\[ (\As,\va) \Bequivk (\Bs,\vb) \; \stackrel{\Delta}{\Longleftrightarrow} \; \forall \vphi \in \BFk. \, [\Ava \models \vphi \IFF \Bvb \models \vphi]. \]

The $k$-round back-and-forth bounded game between $(\As,\va)$ and $(\Bs,\vb)$ is played in the same fashion as the $k+m$-round  Ehrenfeucht-\Fraisse game, with two additional provisos:
\begin{itemize}
\item At round $i : 1 \leq i \leq m$, Spoiler must either play $a_i$ from the tuple $\va$, and Duplicator must respond with $b_i$ from $\vb$; or $b_i$ from the tuple $\vb$, and Duplicator must respond with $a_i$ from $\va$.
\item At round $j>m$, if Spoiler plays a move $a_{j} \in A$ then, for some $i <j$, $\EA(a_i,a_j)$ for some transition relation $E$; while if Spoiler plays a move $b_{j} \in B$ then, for some $i<j$, $\EB(b_i,b_j)$.
\end{itemize}
\begin{proposition}
\label{Bequivgameprop}
Let $(\sg, \sgT)$ be a finite bounded vocabulary.
For all  $(\As,\va)$, $(\Bs,\vb)$ in $\CSmT$, the following are equivalent:
\begin{itemize}
\item $(\As,\va) \Bequivk (\Bs,\vb)$
\item Duplicator has a winning strategy for the $k + m$-round back-and-forth bounded game between $(\As,\va)$ and $(\Bs,\vb)$.
\end{itemize}
\end{proposition}

\begin{theorem}
Let $(\sg, \sgT)$ be a finite bounded vocabulary.
For all  $(\As,\va)$, $(\Bs,\vb)$ in $\CSmT$:
\[  (\As,\va) \Bequivk (\Bs,\vb) \; \IFF \;  (\As,\va) \eqbBk (\Bs,\vb) . \]
\end{theorem}

\subsection{Bounded Counting Quantifier Logic}
As~$\BF$ is a fragment of first order logic, it is natural to consider its extension to a logic~$\BC$ with counting quantifiers, and the relationship to the comonad~$\Bk$. The appropriate model comparison game is a $k + m$-round bounded bijection game, restricting the standard bijection game~\cite{hella1996logical} for first order logic with counting quantifiers, to respect the transition structure. 

Let $(\sg, \sgT)$ be a finite bounded vocabulary, and~$\Ava$ and~$\Bvb$ in~$\CSmT$.
For rounds~$1 \leq i \leq m$ Spoiler must play~$\va_i$, and Duplicator must respond with~$\vb_i$. In round~$i > m$, play proceeds as follows:
\begin{itemize}
\item Duplicator chooses a bijection~$f$ from the subset of elements of~$\As$ that can be seen from some previous position~$a_j$, to the corresponding subset of~$\Bs$. If these subsets have different cardinalities, there is no such bijection, and Spoiler wins.
\item Spoiler picks an element~$a_i$ in~$\As$, and we take~$b_i$ to be~$f(a_i)$.
\end{itemize}
The winning condition for Duplicator is that the correspondence $a_i \mapsto b_i$ is a partial isomorphism from $\As$ to $\Bs$.
\begin{proposition}
Let~$(\sg, \sgT)$ be a finite bounded vocabulary, for all finite~$\Ava$ and~$\Bvb$ in~$\CSmT$, the following are equivalent:
\begin{itemize}
    \item $\Ava \BCequivk \Bvb$.
    \item Duplicator has a winning strategy in the~$m + k$-round bounded bijection game.
\end{itemize}
\end{proposition}
We shall write~$\Ava \eqcBk \Bvb$ for isomorphism in the Kleisli category of the relative comonad~$\Bkp$. The correspondence between isomorphisms in the Kleisli category and winning bounded bijection game strategies follows from an almost identical argument to that given for first order logic in~\cite{abramsky2021relating}.
\begin{proposition}
Let~$(\sg, \sgT)$ be a finite bounded vocabulary. The following are equivalent for finite $\Ava$ and~$\Bvb$ in~$\CSmT$:
\begin{itemize}
    \item $\Ava \eqcBk \Bvb$.
    \item Duplicator has a winning strategy in the~$m + k$-round bounded bijection game.
\end{itemize}
\end{proposition}
\begin{theorem}
Let $(\sg, \sgT)$ be a finite bounded vocabulary.
For all finite $(\As,\va)$, $(\Bs,\vb)$ in $\CSmT$:
\[  \Ava \BCequivk \Bvb \; \IFF \;  \Ava \eqcBk \Bvb. \]
\end{theorem}

\section{Semantic characterization of the bounded fragment}
\label{semcharsec}

We shall now prove a semantic characterisation of the bounded fragment in terms of invariance under generated submodels.
A result of this form is already known \cite{feferman1966persistent,areces2001hybrid}, however there are several novel features in our account:
\begin{itemize}
\item The previous results are for general (possibly infinite) structures, using tools from infinite model theory. We will give a uniform proof, which applies both to general structures, and to the finite case, which, as for the van Benthem-Rosen characterisation of basic modal logic in terms of bisimulation invariance \cite{van1983modal,rosen1997modal}, is an independent result.
\item Our proof follows similar lines to the uniform proof by Otto of the van Benthem-Rosen Theorem \cite{otto2004elementary}. In particular, we use constructive arguments based on model comparison games, rather than model-theoretic constructions.  However, a key property used in his proof no longer holds for the bounded fragment, so the argument has to take a different path. 
\item We also identify a key combinatorial lemma, implicit in \cite{otto2004elementary}, which we call the Workspace Lemma, and give a careful proof at the natural level of generality, in the setting of metric spaces.
\end{itemize}

\subsection{Comonadic aspects}
Comonadic semantics have now been given for a number of important fragments of first-order logic: the quantifier rank fragments, the finite variable fragments, the modal fragment, and guarded fragments. In the landscape emerging from these constructions, some salient properties have come to the fore. These are properties which a comonad,  arising from an \emph{arboreal cover} in the sense of  \cite{abramsky2021arboreal}, may or may not have:
\begin{itemize}
\item The comonad may be \emph{idempotent}, meaning that the comultiplication is a natural isomorphism. Idempotent comonads correspond to \emph{coreflective subcategories}, which form the Eilenberg-Moore categories of these comonads. The modal comonads $\Mk$ are idempotent. The corresponding coreflective subcategories are of those modal structures which are tree-models to depth $k$ \cite{abramsky2021relating}.
\item The comonad $C$ may satisfy the following property: for each structure $\As$, $C \As \lrarr^{C} \As$, where $\lrarr^{C}$ is the back-and-forth equivalence associated with $C$.
We shall call this the \emph{bisimilar companion} property.
Note that an idempotent comonad, such as $\Mk$, will automatically have this property. The guarded comonads $\mathbb{G}_k$ from \cite{abramsky2021comonadic} are not idempotent, but have the bisimilar companion property, which is thus strictly weaker. 
\item Finally, the comonads $\Ek$ and $\Pk$ have neither of the above properties. Unlike the modal and guarded fragments, the quantifier rank and finite variable fragments cover the whole of first-order logic, so we call these comonads \emph{expressive}. 
\end{itemize}
Thus we have a strict hierarchy of comonads in the arboreal categories framework: 
\begin{center}
idempotent $\Rightarrow$ bisimilar companions $\Rightarrow$ arboreal.
\end{center}
This hierarchy is correlated with tractability: the modal and guarded fragments are decidable, and have the tree-model property \cite{vardi1997modal,gradel1999modal}, while the expressive fragments do not.
We can regard these observations as a small first step towards using structural properties of comonadic semantics to classify logic fragments and their expressive power.
In \cite{abramsky2021arbapps}, idempotence is used to give simple, general proofs of homomorphism preservation theorems for counting quantifier fragments, with an application to graded modal logic; while the bisimilar companion property is used to give a general, uniform Otto-style proof of van Benthem-Rosen theorems.

As we have already remarked, the bounded and hybrid comonads are closer to the Ehrenfeucht-\Fraisse comonads $\Ek$ than to the modal comonads $\Mk$. Indeed, $\Hk$ and $\Bk$ are neither idempotent, nor have the bisimilar companion property. On the tractability side, they are not decidable \cite{areces2001hybrid}.  At the same time, they are not fully expressive for first-order logic, thus refining the above hierarchy.

Otto's proof of the van Benthem-Rosen theorem in \cite{otto2004elementary} uses the bisimilar companion property. This is made explicit in the account given in \cite{abramsky2021arbapps}.
Because $\Hk$ and $\Bk$ do not have this property, we shall use a different comonad in our invariance proof for the bounded and hybrid fragments.

We shall use the \emph{reachability comonad} $\RR$, defined on $\CSmT$ as follows. Given a structure $(\As,\va)$, let $\to$ be the union of all $\EA$, $E \in \sgT$. The universe of $\RR \As$ is the set of all elements $a$ of $A$ such that there is a path $a_i \to^* a$, for some $i: 1 \leq i \leq m$. Then $\RR \As$ is the corresponding induced substructure of $(\As,\va)$. The counit is the inclusion map, while coextension is the identity operation on morphisms, $h^* = h$. The fact that $h$ is a $\sg$-homomorphism implies that paths are preserved, so this is well defined. It is easily verified that $\RR$ is an idempotent comonad. The corresponding coreflective subcategory of $\CSmT$ is the full subcategory of structures which are reachable from the initial elements.
For each $k>0$, there is a sub-comonad $\Rk$ of elements which are $k$-reachable.

We can use this comonad to state the invariance property of interest. We say that a first-order sentence $\vphi$ is \emph{invariant under generated substructures} if for all $(\As,\va)$ in $\CSmT$:
\[ \Ava \models \vphi \; \IFF \; \RR \Ava \models \vphi . \]
It is \emph{invariant under $k$-generated substructures} if for all $(\As,\va)$ in $\CSmT$:
\[ \Ava \models \vphi \; \IFF \; \Rk \Ava \models \vphi . \]

Note that we are regarding $\va$ as the denotations of constants $c_1,\ldots,c_m$ which may occur in the sentence $\vphi$. We could equally well think of $\vphi$ as containing free variables $v_1,\ldots,v_m$, for which $\va$ provides an assignment.
\begin{proposition}
\label{invextprop}
If $\vphi$ is invariant under $k$-generated substructures, and $k \leq l$, then $\vphi$ is invariant under $l$-generated substructures.
\end{proposition}
\begin{proof}
Note firstly that $\Rk \Rl \Ava = \Rk \Ava$. Now if $\vphi$ is invariant under $k$-generated substructures, 
\[ \Ava \models \vphi \;  \IFF \; \Rk \Ava \models \vphi \; \IFF \; \Rk \Rl \Ava \models \vphi \; \IFF \; \Rl \Ava \models \vphi . \]
\end{proof}
We can now state our main result.
\begin{theorem}[Characterisation Theorem]
For any first-order sentence $\vphi$, the following are equivalent:
\begin{enumerate}
\item $\vphi$ is invariant under $k$-generated substructures for some $k>0$.
\item $\vphi$ is equivalent to a sentence $\psi$ in the bounded fragment. 
\end{enumerate}
\end{theorem}
Note that this theorem has two versions, depending on the ambient category $\CC$ relative to which equivalence is defined:
\[ \forall (\As,\va) \in \CC. \, \Ava \models \vphi \; \IFF \; \Ava \models \psi . \]
The first version, for general models, takes $\CC = \CSmT$. The second, for finite models, takes $\CC = \CSmTf$, the full subcategory of finite structures.
Neither of these two versions implies the other.
Following Otto \cite{otto2004elementary}, we aim to give a uniform proof, valid for both versions.

\subsection{Proof of the Characterisation Theorem}

Firstly, since any sentence can only use a finite vocabulary, we can assume without loss of generality in what follows that $\sg$ is finite.
This implies that up to logical equivalence, the fragment $\BFk$ is finite.

Given a formula $\vphi$, we write $\Mod(\vphi) := \{ \Ava \mid (\As,\va) \models \vphi \}$.
We shall use the following variation of a standard result.
\begin{lemma}[Definability Lemma]
For each $k>0$ and structure $\Ava$, there is a sentence $\hAvk \in \BFk$ such that, for all $(\Bs,\vb)$:
\[ (\As,\va) \Bequivk (\Bs,\vb) \; \IFF \; (\Bs,\vb) \models \hAvk . \]
\end{lemma}
This says that $[(\As,\va)]_{\Bequivk} = \Mod(\hAvk)$. Since $\Bequivk$ has finite index, this implies that if $\Mod(\vphi)$ is saturated under $\Bequivk$, $\vphi$ is equivalent to a finite disjunction $\bigvee_{i=1}^n \hAik$, and hence to a formula in $\BFk$.
Thus to prove the Characterisation theorem, it is sufficient to prove that $\Mod(\vphi)$ is saturated under $\Bequivk$ for some $k$ whenever $\vphi$ is invariant under generated substructures.

We recall the standard disjoint union of structures, $\As + \Bs$. This is the coproduct in $\CS$. We say that a sentence $\vphi$ is \emph{invariant under disjoint extensions} if for all $(\As,\va)$, $\Bs$:
\[ (\As,\va) \models \vphi \; \IFF \; (\As + \Bs, \va) \models \vphi . \]
Note that $\Rk (\As + \Bs, \va) = \Rk (\As, \va)$. Hence the following is immediate:
\begin{lemma}
\label{invgsdelemm}
Invariance under $k$-generated substructures implies invariance under disjoint extensions.
\end{lemma}

\subsubsection*{The Workspace Lemma}
A key step in the argument is a general result we call the Workspace Lemma. A special case of this is implicit in \cite{otto2004elementary}. 
To state this result, we need to consider another comonad on $\CSmT$.
Given a structure $\As$, we can define a metric on $A$ valued  in the extended natural numbers $\Nat \cup \{ \infty\}$, given by the path distance in the Gaifman graph $\Gf(\As)$ \cite{Libkin2004}. 
We set $d(a,b) = \infty$ if there is no path between $a$ and $b$. We write $A[a;k]$ for the closed ball centred on $a$, and extend this to tuples $\va$ by $A[\va;k] := \bigcup_{i=1}^n A[a_i;k]$. Given $\Ava$, we define $\Sk \Ava$ to be $(\Ask,\va)$, where $\Ask$ is the substructure of $\As$ induced by $A[\va;k]$. This defines an idempotent comonad on $\CSmT$, in similar fashion to $\Rk$. Note that, as for $\Rk$, $\Sk \Ava = \Sk (\As + \Bs, \va)$.

We can now state the Workspace Lemma.
Note that $\equiv_q$ is elementary equivalence up to quantifier rank $q$.
\begin{lemma}[Workspace Lemma]
Given $(\As,\va)$ and $q>0$, there is a structure $\Bs$ such that $(\As + \Bs,\va) \equiv_q  (\Ask + \Bs, \va)$, where $k = 2^q$. 
Moreover, $|B| \leq 2q|A|$. Hence if $\As$ is finite, so is $\Bs$.
\end{lemma}
The proof of this result is deferred to the next subsection.

We can define a partial order on the objects of $\CSmT$ by $\Ava \strpo \Bvb$ iff $\As$ is an induced substructure of $\Bs$, and $\va = \vb$. An endofunctor on $\CSmT$ is \emph{monotone} if the object part is monotone with respect to this order.

\begin{proposition}\label{RkSkprop}
For each structure $\Ava$ in $\CSmT$, we have $\Rk \Sk \Ava = \Rk \Ava$.
\end{proposition}
\begin{proof}
For the left-to-right inclusion, $\Sk \Ava \strpo \Ava$, and $\Rk$ is monotone.
For the converse inclusion, each directed $\sgT$-path in $\As$ is a path in the Gaifman graph of $\As$ as a $\sg$-structure. Hence $\Rk \Ava \strpo \Sk \Ava$. Moreover, $\Rk$ is idempotent and monotone, which yields the result.
\end{proof}

We shall need two more lemmas.
\begin{lemma}
\label{Reqlemm}
For all $k,m>0$, $\Ava \Bequivm \Bvb \IMP \Rk \Ava \Bequivm \Rk \Bvb$.
\end{lemma}
\begin{proof}
We use Proposition~\ref{Bequivgameprop}. Assume that we have a winning strategy $S$ for Duplicator in the bounded game between $\Ava$ and $\Bvb$.
This cuts down to a winning strategy in the game between $\Rk \Ava$ and $\Rk \Bvb$. Indeed, suppose we have reached a position $(s,t)$ and Spoiler makes a valid move in the $k$-reachable subset of $A$. Duplicator's response according to $S$ must preserve the partial isomorphism winning condition, and hence must be in the $k$-reachable subset of $B$. The argument when Spoiler moves in $B$ is symmetric.
\end{proof}

\begin{lemma}
\label{Beqlemm}
For all $k,q>0$, $\Rk \Ava \Bequivkq \Rk \Bvb \IMP \Rk \Ava \equiv_q \Rk \Bvb$.
\end{lemma}
\begin{proof}
Again, we use Proposition~\ref{Bequivgameprop}. Assume that we have a winning strategy $S$ for Duplicator in the bounded game of length $kq$ between 
$\Rk \Ava$ and $\Rk \Bvb$. We shall use this to define a winning strategy $S'$ in the Ehrenfeucht-\Fraisse game of length $q$ between $\Rk \Ava$ and $\Rk \Bvb$, in such a way that for each position $(s,t)$ reachable following $S'$, there is a ``covering position'' $(s^*,t^*)$ reachable following $S$, where $s$ is a subsequence of $s^*$ and $t$ is the corresponding subsequence of $t^*$. The partial isomorphism winning condition for $S$ then implies the corresponding condition for $S'$. The length of $(s^*,t^*)$ will be at most $k$ times the length of $(s,t)$. 

Suppose, arguing inductively, we have reached $(s,t)$ following $S'$, with covering position $(s^*,t^*)$. If Spoiler makes a move $a \in \Rk \Ava$ then $a$ must be $k$-reachable, so we can extend $s^*$ by Spoiler making at most $k$ moves in $\Rk \Ava$ ending with $a$, and extend $t^*$ by following the responses to these moves in $\Rk \Bvb$ according to $S$. We can use the final move in this extension as the response by $S'$ to $a$. We proceed symmetrically if Spoiler moves in $\Rk \Bvb$.
\end{proof}
We are now ready to prove the main part of the Characterisation theorem.
\begin{proposition}
\label{invdebfprop}
If $\vphi$ is invariant under $k$-generated substructures for some $k>0$, and has quantifier rank $q$, then it is equivalent to a formula $\psi$ in the bounded fragment with quantifier rank $\leq q2^m$, where $m = \max(k,q)$.
\end{proposition}
\begin{proof}
Let $r = 2^m$.
Suppose that (1) $\Ava \models \vphi$, and (2) $\Ava \Bequivrq \Bvb$. We must show that $\Bvb \models \vphi$.
Applying the Workspace Lemma twice, let $\Cs$, $\Ds$ be such that
\[ (3) \quad (\As + \Cs,\va) \equiv_q  (\Asr + \Cs, \va) \]
\[ (4) \quad (\Bs + \Ds,\vb) \equiv_q  (\Bsr + \Ds, \vb) \]
From (2), applying lemmas~\ref{Reqlemm} and~\ref{Beqlemm}, we have
\[ (5) \quad \Rr \Ava \equiv_q \Rr \Bvb \]
Since $\vphi$ is invariant under $k$-generated substructures (abbreviated as $\IGSk$), by Lemma~\ref{invgsdelemm}, it is invariant under disjoint extensions (abbreviated as $\IDE$).
\[ \begin{array}{lcll}
\Ava \models \vphi & \IMP & (\As + \Cs,\va) \models \vphi & \IDE \\
& \IMP & (\Asr + \Cs, \va) \models \vphi & (3) \\
& \IMP & \Sr (\As, \va) \models \vphi & \IDE \\
& \IMP & \Rr \Sr (\As, \va) \models \vphi &  \IGSk, \mbox{Proposition~\ref{invextprop}} \\
& \IMP & \Rr (\As, \va) \models \vphi &  \mbox{Proposition~\ref{RkSkprop}} \\
& \IMP & \Rr (\Bs, \vb) \models \vphi & (5) \\
& \IMP & \Rr \Sr (\Bs, \vb) \models \vphi &  \mbox{Proposition~\ref{RkSkprop}} \\
& \IMP & \Sr (\Bs, \vb) \models \vphi & \IGSk, \mbox{Proposition~\ref{invextprop}} \\
& \IMP & (\Bsr + \Ds, \vb) \models \vphi & \IDE \\
& \IMP & (\Bs + \Ds, \vb) \models \vphi & (4) \\
& \IMP & (\Bs, \vb) \models \vphi & \IDE \\
\end{array}
\]
\end{proof}

\begin{question}
In his proof of the van Benthem-Rosen Theorem, Otto establishes an exponential succinctness gap between first-order logic and basic modal logic. A bisimulation-invariant first order formula of quantifier rank $q$ has a modal equivalent of modal depth $\leq 2^q$. He shows that this is optimal. In our case, we have a gap of $q2^m$. Is this optimal for the bounded fragment?
\end{question}

\begin{proposition}
\label{reachprop}
If $\psi$ is a formula in $\BFk$, then it is invariant under $k$-generated substructures.
\end{proposition}
\begin{proof}
A straightforward induction on the syntax.
\end{proof}

Since any formula in the bounded fragment is in $\BFk$ for some $k$,
combining Proposition~\ref{invdebfprop} and Proposition~\ref{reachprop} we obtain a proof of the Characterisation Theorem.

\subsection{Discussion}

In \cite[Theorem 3.7]{areces2001hybrid}, the following characterization result is proved.
\begin{theorem}
Given a first-order formula $\vphi$, the following are equivalent:
\begin{itemize}
\item $\vphi$ is invariant under generated substructures
\item $\vphi$ is equivalent to a formula in the bounded fragment.
\end{itemize}
\end{theorem}
This is proved relative to general (infinite) models, using methods of infinite model theory. Combining this result with Proposition~\ref{reachprop} yields the following:
\begin{proposition}
If a first-order formula is invariant under generated substructures, it is invariant under $k$-generated substructures for some $k$.
\end{proposition}

\begin{question}
Can we give a constructive proof of this result, with a bound for $k$?
Can we give a uniform proof, valid also in the finite model case?
\end{question}

\subsection{Proof of the Workspace Lemma}

The Workspace Lemma is extracted from  \cite{otto2004elementary}, and we shall largely follow the construction sketched there. 
However, there are some important differences:
\begin{itemize}
\item The argument in \cite{otto2004elementary} makes crucial use of the assumption that the formula under consideration is bisimulation-invariant, which allows us to assume that we are dealing with tree models, whereas we make no such assumption here.
\item The distance used in the argument in \cite{otto2004elementary} is \emph{directed path distance} with respect to the modal transition relation. This is not symmetric, \ie it defines a quasi-metric rather than a metric. This is not sufficient to carry through the argument given below.
\end{itemize}
As the details are somewhat intricate, we shall spell them out.

We wish to construct a winning strategy for the $q$-length Ehrenfeucht-\Fraisse game between $\Ava$ and $\Sk \Ava$. The obvious problem is what to do when Spoiler plays in a part of $\Ava$ which is not $k$-reachable, and hence has no counterpart in $\Sk \Ava$. To address this problem, we form a structure $\Bs$ comprising the disjoint union of sufficiently many copies of $\Ava$ and $\Sk \Ava$ -- the \emph{workspace}, and play the game between $(\As + \Bs,\va)$ and $(\Ask + \Bs, \va)$. This allows us to split each play into a combination of a number of copy-cat strategies, each played between two copies of the same structure. These copy-cat strategies are kept well-separated, so they can be combined into a single winning strategy for Duplicator in the overall game.
The subtlety in the construction is that the splitting into disjoint sub-strategies has to be determined dynamically as the game proceeds, in response to Spoiler moves.

This construction is very much in the spirit of game semantics \cite{abramsky1994games,hyland2000full,abramsky2000full}. Combinations of copycat strategies are the typical denotations of logical proofs in game semantics. Adding workspace in order to define the copycat strategies is \emph{logic plus resources}.

The core of the construction and its properties can be developed efficiently in the general setting of metric spaces. We shall do this first, before returning to relational structures to prove the Workspace Lemma.

\subsubsection{Strategies in metric spaces}
Metrics are taken to be valued in the extended non-negative reals $[0,\infty]$.
We shall form disjoint sums  $M+N$ of metric spaces, with $d(x,y) = \infty$ for $x \in M$ and $y \in N$.
More generally, we form sums $M = \sum_{i \in I} M_i$ of families of metric spaces. We refer to the $M_i$ as \emph{summands} of~$M$.

We write $B(x;r)$ for the open ball in a metric space centred on $x$, with radius $r$. The closed ball is written as $B[x;r]$.
This notation is extended to tuples $\vx \in M^n$ by $B[\vx;r] := \bigcup_{i=1}^n B[x_i;r]$.
The distance between subsets $S, T$ of a metric space is defined by $d(S,T) \coloneqq \inf \{ d(x,y) \mid x \in S, \, y \in T \}$.
We will rely on context to indicate which ambient metric space  is intended.

Given sets $X$, $Y$, and tuples $\vx \in X^m$, $\vy \in Y^m$, a \emph{strategy of length $k$} between $(X,\vx)$ and  $(Y,\vy)$ is a set $\str \subseteq \bigcup_{0 \leq j \leq k} X^{m+j} \times Y^{m+j}$ of pairs $(s,t)$, where $s$ is a sequence of elements of $X$, and $t$ is a sequence of elements of $Y$ of the same length. This must satisfy the following conditions:
\begin{itemize}
\item $(\vx,\vy) \in \str$.
\item If $(s,t) \in \str$, with common length $j < m+k$, then 
\begin{enumerate}
\item $\forall x \in X. \, \exists y \in Y. \, (sx, ty) \in \str$
\item $\forall y \in Y. \, \exists x \in X. \, (sx, ty) \in \str$.
\end{enumerate}
\end{itemize}

Now, suppose we are given a metric space $M$, a tuple $\va = \langle a_1, \ldots , a_m \rangle \in M^m$, and a positive integer $q$. Set $\ell \coloneqq 2^q$ and, for each $k\in\{0,\ldots, q\}$, $\ell_k \coloneqq 2^{q-k}$. Let $N\coloneqq B[\va;\ell]$.
Observe that
\[
\{a_1,\ldots,a_m\}=B[\va;\ell-\ell_0]\subseteq B[\va;\ell-\ell_1]\subseteq \cdots \subseteq B[\va;\ell-\ell_q]\subseteq N.
\]

We define $P \coloneqq q\cdot M + q \cdot N$, the disjoint sum of $q$ copies of $M$ and $q$ copies of $N$. Thus this is a disjoint sum of $2q$ spaces.
We define $M' \coloneqq M+P$, $N' \coloneqq N+P$. Note that $M'$ has $q+1$ subspaces which are  copies of $M$, and $q$ subspaces which are copies of $N$, and similarly, with the roles of $M$ and $N$ reversed, for $N'$.
For each copy of $M$ occurring as a summand (and hence as a subspace) $S$ of $M'$, and as a summand $T$ of $N'$, there is a canonical isomorphism $\rho^{S,T} \colon S \cong T$.
Similarly, for each copy  of $N$ occurring as a summand  $U$ of $M'$, and as a summand $V$ of $N'$, there is a canonical isomorphism $\rho^{U,V} \colon U \cong V$.

Intuitively, our aim is to prove that there is a strategy $\str$ of length $q$ between $(M',\va)$
and $(N',\va)$ which resolves into a number of disjoint components, which are copy-cat strategies  played between isomorphic subspaces of $M'$ and $N'$. Moreover, these components are well-separated, and do not interfere with each other.

This can be expressed formally as follows.
\begin{proposition}\label{pr:strategy-metric}
There is a strategy $\str$ of length $q$ between $(M',\va)$ and $(N',\va)$, such that for all $(s,t) \in \str$, where $s = \langle a_1, \ldots , a_m, \ldots a_{m+k} \rangle$,
$t = \langle b_1, \ldots , b_{m+k} \rangle$, with $0 \leq k \leq q$, the following conditions hold:
\begin{enumerate}
\item There is a bipartition $C_0 \, \sqcup \, C_1$ of $\{ a_1, \ldots , a_{m+k} \}$ and a bipartition $D_0 \, \sqcup \, D_1$ of $\{ b_1, \ldots , b_{m+k} \}$.
We have $\{a_1,\ldots,a_m\} \subseteq C_0 \cap D_0$.
\item $C_0 \cup D_0 \subseteq B[\va;\ell - \ell_k]$, and for all $i: 1 \leq i \leq m+k$, $a_i \in C_0$ iff $b_i \in D_0$.
\item Let $(s_0,t_0)$ be the subsequences of $(s,t)$ obtained by selecting those elements in $C_0$ and $D_0$ respectively. Then $s_0=t_0$.
\item We have $d(C_0,C_1) > \ell_k$, and $d(D_0,D_1) > \ell_k$.
\item For all $i\in\{m+1,\ldots,m+k\}$, if $a_i \in C_1$, and $b_i \in D_1$, then $a_i \in S$ and $b_i \in T$, where $S$ and $T$ are  summands of $M'$ and $N'$ respectively of the same type, \ie~both are copies of $M$ or both are copies of $N$, and $b_i = \rho_{i}(a_i)$, where $\rho_i = \rho^{S,T}$.
\item For all $i,j\in\{m+1,\ldots, m+k\}$, if $a_i, a_j \in C_1$, and $b_i,b_j \in D_1$, then $d(a_i,a_j) \leq \ell_k \vee d(b_i,b_j) \leq \ell_k \; \Rightarrow \; \rho_i = \rho_j$.
\end{enumerate}
\end{proposition}

\begin{proof}
By induction on the length of plays $(s,t)$. We use $k$ as the induction parameter, $0 \leq k < q$. 
We refer to the induction hypothesis at stage $k+1$ as $\Inv(k)$.
Note that $\ell_k = 2\ell_{k+1}$, i.e.~the distance shrinks by $1/2$ at each step.

The base case $k=0$ is immediate.
For the inductive step at $k+1$, given $(s,t) \in \str$, let $p = m+k+1$. Given $a_{p} \in M'$, we distinguish between three possible cases:\footnote{Note that Cases I and II below are mutually exclusive by $\Inv(k)(4)$. This requires the distance function to be symmetric, so undirected distance must be used.}

Case I: for some $i\in\{1,\ldots, m+k\}$,  $d(a_i, a_{p}) \leq \ell_{k+1}$ and $a_i \in C_0$. 
By $\Inv(k)(2)$ and the triangle inequality, for some $j:1\leq j\leq m$, $d(a_{j},a_{p}) \leq \ell - \ell_k + \ell_{k+1} = \ell - \ell_{k}$. That is, $a_{p}\in B[\va,\ell - \ell_{k}]$.
By $\Inv(k)(3)$, $s_0=t_0$. 
We extend $\str$ by adding $(sa_{p},ta_{p})$ and extend the bipartition by adding $a_{p}$ to both $C_0$ and $D_0$.
Thus (1)--(3) are maintained. It remains to verify (4). Suppose that, for some $j$ with $a_j \in C_1$ and $b_j \in D_1$, we have $d(a_{p},a_j) \leq \ell_{k+1}$ or $d(a_{p},b_j) \leq \ell_{k+1}$. Assuming the former, by the triangle inequality, $d(a_i,a_j) \leq 2 \ell_{k+1} = \ell_k$, contradicting $\Inv(k)(4)$. A similar argument applies in the latter case.

Case II: for some $i\in\{1,\ldots, m+k\}$,  $d(a_i, a_{p}) \leq \ell_{k+1}$ and $a_i \in C_1$. 
In this case, by $\Inv(k)(5)$, $b_i \in D_1$ and  $b_i = \rho_{i}(a_i)$. We must have $a_{p}$ in the same summand of $M'$ as $a_i$, since  $d(a_i, a_{p}) < \infty$.
We define $b_{p} \coloneqq \rho_{p}(a_{p})$, where $\rho_{p} := \rho_{i}$. We add $a_{p}$ to $C_{1}$ and $b_{p}$ to $D_1$, and extend $\str$ by adding $(sa_{p},tb_{p})$. 
Note that for any $j$, if $d(a_j, a_{p}) \leq \ell_{k+1}$, then $d(a_i,a_j) \leq \ell_k$, so by $\Inv(k)(6)$, $\rho_j = \rho_i = \rho_{p}$; thus (6) is maintained.
It remains to verify (4). The argument is identical to that given in the previous case, with the roles of $C_0$ and $C_1$ (resp. $D_0$ and $D_1$) interchanged.

Case III: for all $i\in\{1,\ldots, m+k\}$, $d(a_i,a_{p}) > \ell_{k+1}$. In this case, let $S$ be the summand of $M'$ such that $a_p \in S$. We choose a summand $T$ of the same type ($M$ or $N$) in $N'$ which has not been used at any previous stage, \ie~$b_j \not\in T$ for any $0 \leq j \leq k$. This is always possible, as there are at least $q$ summands of the required type  in $N'$. We define $b_{p} \coloneqq \rho_{p}(a_{p})$, where $\rho_{p} := \rho^{S,T}$. We add $a_{p}$ to $C_{1}$ and $b_{p}$ to $D_1$, and extend $\str$ by adding $(sa_{p},tb_{p})$. Since $d(b_j,b_{p}) = \infty$ for any $1 \leq j \leq m+k$, (6) is maintained.

A symmetric argument shows that we can extend $\str$ given any ${b_{p} \in N'}$.
\end{proof}

\subsubsection{Completing the proof of the Workspace Lemma}

We now return to relational structures. Given a structure $\As$, we have a metric on $A$ valued in the extended natural numbers $\Nat \cup \{ \infty\}$, given by the path distance in the Gaifman graph $\Gf(\As)$.
We set $d(x,y) = \infty$ if there is no path between $x$ and $y$.

Applying Proposition~\ref{pr:strategy-metric}, we obtain a strategy of length $q$ between $(\As + \Cs,\va)$ and $(\Ask + \Cs,\va)$, where $\Cs := q \cdot \As + q \cdot \As_k$.
It remains to verify that this is a winning strategy for Duplicator in the Ehrenfeucht-\Fraisse game. This amounts to showing that the strategy satisfies the partial isomorphism winning condition.
\begin{lemma}\label{pisolemm}
Let $\alpha$, $\beta$ be partial isomorphisms from $\Ava$ to $\Bvb$. If $d(\dom(\alpha), \dom(\beta)) > 1$ and $d(\ran(\alpha), \ran(\beta)) > 1$, then $\alpha \cup \beta$ is a partial isomorphism.
\end{lemma}
\begin{proof}
The conditions on distances imply firstly that the domains and ranges of the partial isomorphisms are disjoint, so their union is a partial bijection.
Moreover, no element of $\dom(\alpha)$ is adjacent to any element of $\dom(\beta)$ in the Gaifman graph, and similarly for $\ran(\alpha)$, $\ran(\beta)$, so no additional requirements of preservation of relations arise in the union.
\end{proof}
We now complete the proof of the Workspace Lemma.
\begin{proof}
Given any play $(s,t)$ following the strategy produced by Proposition~\ref{pr:strategy-metric}, we can partition it into subsequences $(s_0,t_0)$, $(s_1,t_1)$ according to parts (1)--(3) of the Proposition. We can further partition $(s_1,t_1)$, where $a_i$, $a_j$ are put in the same element of the partition iff $\rho_i = \rho_j$. Each of these subsequences produces a relation $a_i \mapsto b_i$. In the case of $(s_0,t_0)$, this  is a subset of the graph of the inclusion $\Ask \hookrightarrow \As$, while in the case of the partition with corresponding isomorphism $\rho^{S,T}$, it is a subset of the graph of this isomorphism. Moreover, by parts (4) and (6) of the Proposition, the distances between the graphs of these partial isomorphisms are $> \ell_q = 1$. Thus by Lemma~\ref{pisolemm}, their union is a partial isomorphism. 
\end{proof}

\section{Hybrid Temporal Logic}

We shall now consider an extension of Hybrid Logic, in which backwards modalities $\Bm$, $\Dm$ are added, with the following first-order translations:
\[ \begin{array}{lcl}
\ST_{x}(\Bm \vphi) & = & \forall y. [ E(y,x) \to \ST_{y}(\vphi)] \\
\ST_{x}(\Dm \vphi) & = & \exists y. [ E(y,x) \wedge \ST_{y}(\vphi)] \\
\end{array}
\]
with the stipulation that $x \neq y$. 
The corresponding extension of the bounded fragment adds bounded quantifiers 
\[ \exists x. \, E(x,t) \wedge \vphi, \qquad \forall x. \,  E(x,t) \to \vphi \]
with the stipulation that $x \neq t$.

This extension is natural for Hybrid Logic; the backwards modalities are standard in temporal logic. The extended logic admits a simple and natural semantic characterization, as we shall now see.

Firstly, we note that in terms of the comonads $\Hk$, this extension is accommodated as follows:
\begin{itemize}
\item $\Hkm (\As,a)$ has as universe the subset of $\Ek (A,a)$ of those sequences $\langle a_0, a_1, \ldots , a_{l}\rangle$ such that $a_0 =a$, and for all $j$ with $0 < j \leq l$, for some $i$, $0 \leq i < j$, 
$\EA(a_{i},a_{j})$ or $\EA(a_{j},a_{i})$ . $\Hkm (\As,a)$ is the induced substructure of $\Ek (\As,a)$ given by this restriction of the universe. 
\end{itemize}
The further development of this comonad proceeds entirely analogously to that of $\Hk$; we shall not elaborate the details.

We can state the semantic characterization theorem for Hybrid Temporal Logic as follows.
\begin{theorem}[Characterisation Theorem for Hybrid Temporal Logic]
For any first-order formula $\vphi(x)$ in the unimodal signature, the following are equivalent:
\begin{enumerate}
\item $\vphi$ is invariant under disjoint extensions.
\item $\vphi$ is equivalent to a formula $\psi$ of Hybrid Temporal Logic under the standard translation. 
\end{enumerate}
\end{theorem}
We may regard invariance under disjoint extensions as a minimal form of locality relative to a given basepoint. Thus this characterization shows that Hybrid Temporal Logic defines the maximal fragment of first-order logic which retains a local character in this sense.

The $k$-round back-and-forth hybrid temporal game between  $(\As,a)$ and $(\Bs,b)$ is defined in the same way as the standard  $k$-round Ehrenfeucht-\Fraisse game, with two additional provisos:
\begin{itemize}
\item At round $0$, Spoiler must either play $a_0 = a$, to which Duplicator must respond with $b_0 = b$; or $b_0 = b$, to which Duplicator must respond with $a_0 = a$
\item At round $j>0$, if Spoiler plays a move $a_{j} \in A$ then, for some $i<j$, $\EA(a_i,a_j)$ or $\EA(a_j,a_i)$ ; while if Spoiler plays a move $b_{j} \in B$ then, for some $i<j$, $\EB(b_i,b_j)$ or $\EB(b_j,b_i)$.
\end{itemize}
The partial isomorphism winning condition ensures that Duplicator is subject to the same constraints.

We then obtain the following analogue of Theorem~\ref{Hequivth}.
\begin{theorem}
\label{Hmequivth}
Let $\sg$ be a finite unimodal vocabulary.
For all  $(\As,a)$, $(\Bs,b)$ in $\CSp$, the following are equivalent:
\begin{enumerate}
\item Duplicator has a winning strategy for the $k$-round back-and-forth hybrid temporal game between  $(\As,a)$ and $(\Bs,b)$.
\item $(\As,a) \Hmequivk (\Bs,b)$.
\end{enumerate}
\end{theorem}

For proving the semantic characterization, the basic idea is that the comonad $\Sk$ plays the analogous role for Hybrid Temporal Logic that $\Rk$ does for Hybrid Logic.
The following are straightforward analogues of Lemmas~\ref{Reqlemm} and~\ref{Beqlemm}.

\begin{lemma}
\label{Hmeqlemm}
For all $k,m>0$, $\Asa \Hmequivm \Bsb \IMP \Sk \Asa \Hmequivm \Sk \Bsb$.
\end{lemma}

\begin{lemma}
\label{HtoEeqlemm}
For all $k,q>0$, $\Sk \Asa \Hmequivkq \Sk \Bsb \IMP \Sk \Asa \equiv_q \Sk \Bsb$.
\end{lemma}

We can now prove an analogue of Proposition~\ref{invdebfprop}.
\begin{proposition}
\label{invdeprop}
If $\vphi$ is invariant under disjoint extensions, and has quantifier rank $q$, then it is equivalent to a formula $\psi$ in Hybrid Temporal Logic with modal depth $\leq q2^q$.
\end{proposition}
\begin{proof}
Let $k = 2^q$.
Suppose that (1) $\Asa \models \vphi$, and (2) $\Asa \Hmequivkq \Bsb$. We must show that $\Bsb \models \vphi$.
Applying the Workspace Lemma twice, let $\Cs$, $\Ds$ be such that
\[ (3) \quad (\As + \Cs,a) \equiv_q  (\Askm + \Cs, a) \]
\[ (4) \quad (\Bs + \Ds,b) \equiv_q  (\Bskm + \Ds, b) \]
From (2), applying lemmas~\ref{Hmeqlemm} and~\ref{HtoEeqlemm}, we have
\[ (5) \quad \Sk \Asa \equiv_q \Sk \Bsb \]
Now
\[ \begin{array}{lcll}
\Asa \models \vphi & \IMP & (\As + \Cs,a) \models \vphi & \IDE \\
& \IMP & (\Askm + \Cs, a) \models \vphi & (3) \\
& \IMP & \Sk (\As, a) \models \vphi & \IDE \\
& \IMP & \Sk (\Bs, b) \models \vphi & (5) \\
& \IMP & (\Bskm + \Ds, b) \models \vphi & \IDE \\
& \IMP & (\Bs + \Ds, b) \models \vphi & (4) \\
& \IMP & (\Bs, b) \models \vphi & \IDE \\
\end{array}
\]
\end{proof}

A straightforward induction on syntax proves:
\begin{proposition}
\label{HTLreachprop}
If $\psi$ is a formula of Hybrid Temporal Logic, then it is invariant under disjoint extensions.
\end{proposition}

Combining Propositions~\ref{invdeprop} and~\ref{HTLreachprop}, we obtain a proof of the Characterisation Theorem for Hybrid Temporal Logic.

\subsection{Connection with Gaifman Locality}

We shall now relate these notions to standard ideas in first-order logic concerning locality and Gaifman's Theorem \cite{gaifman1982local}.
We recall that Gaifman distance with respect to a finite relational vocabulary is first-order definable. For each first-order formula $\vphi(\vx)$ and $k>0$, there is a formula $\vphivk$ defined by structural induction on $\vphi$. The only non-trivial case is for quantifiers:
\[ (\exists y.\, \psi)^{(\vx,k)} \, := \, \exists y. \, (d(\vx, y) \leq k  \wedge \psivk ) . \] 
These \emph{basic $k$-local formulas about $\vx$} satisfy the following property \cite[p. 31]{ebbinghaus1999finite}:
\[ (*) \;\; \Ava \models \vphivk \; \IFF \; \Sk \Ava \models \vphi . \]
We say that $\vphi$ is \emph{$\Sk$-invariant} if for all $\Ava$,
\[ \Ava \models \vphi \; \IFF \; \Sk \Ava \models \vphi . \]
\begin{theorem}
Let $\vphi(\vx)$ be a first-order formula of quantifier rank $q$, and let $k = 2^q$. The following are equivalent:
\begin{enumerate}
\item $\vphi$ is invariant under disjoint extensions.
\item $\vphi$ is $\Sk$-invariant.
\item $\vphi$ is equivalent to $\vphivk$.
\end{enumerate}
\end{theorem}
\begin{proof}
The equivalence of (2) and (3) follows directly from (*). Since $\Sk (\As + \Bs, \va) = \Sk (\As, \va)$, (2) implies (1).
Finally, we can use the Workspace Lemma as in the proof of Proposition~\ref{invdebfprop} to show that (1) implies (2):
\[ \Ava \models \vphi \; \IFF \; (\As + \Cs,\va) \models \vphi  \; \IFF \; (\Ask + \Cs,\va) \models \vphi  \; \IFF \; \Sk \Ava \models \vphi  . \]
\end{proof}
\bibliography{hybrid}

\appendix

\section{Bounded Fragment Model Comparison Games}
We are not aware of previous explicit descriptions of the model comparison games for the bounded fragment. This section develops the necessary results, adapting proofs for full first order logic appearing in the literature.

\subsection{The Ehrenfeucht-\Fraisse Game}
The proof is an adaptation of the argument used in the proof of~\cite[Theorem 3.18]{Libkin2004}. 
\begin{definition}[Bounded Back and Forth Relations]
We inductively define the following relations:
\begin{itemize}
    \item $\Ava \bfe{0} \Bvb$ if both structures satisfy the same atomic sentences.
    \item $\Ava \bfe{k + 1} \Bvb$ if the following two conditions hold for every constant~$c \in \sg$ and binary relation symbol~$E$:
    \begin{description}
        \item[forth]: For every~$a \in \As$ such that~$E^\As(c^\As,a)$ there exists~$b \in \Bs$ such that~$E^\Bs(c^\Bs,b)$ and~$(\As, \va,a) \bfe{k} (\Bs, \vb, b)$.
        \item[back]: For every~$b \in \Bs$ such that~$E^\Bs(c^\Bs,b)$ there exists~$a \in \As$ such that~$E^\As(c^\As,a)$ and~$(\As, \va, a) \bfe{k} (\Bs, \vb, b)$.
    \end{description}
\end{itemize}
\end{definition}
We shall write~$\Ava \efe{k} \Bvb$ if Duplicator has a winning strategy in the~$m + k$-round bounded Ehrenfeucht-\Fraisse game between~$\As$ and~$\Bs$.
\begin{theorem}[Bounded Ehrenfeuct-\Fraisse]
For two~$\sg$-structure~$\Ava$ and~$\Bvb$ the following are equivalent:
\begin{enumerate}
    \item \label{lab:loge} $\Ava \Bequivk \Bvb$.
    \item \label{lab:efe} $\Ava \efe{k} \Bvb$.
    \item \label{lab:bfe} $\Ava \bfe{k} \Bvb$.
\end{enumerate}
\end{theorem}
\begin{proof}
We proceeding by induction on~$k$. The base case~$k = 0$ is immediate from the definitions.

For the inductive step, we first show~\ref{lab:bfe} implies~\ref{lab:efe}. Assume~$\Ava \bfe{k + 1} \Bvb$, and that Spoiler plays~$a \in \As$ with~$E^\As(c^\As,a)$. Then by assumption there exists~$b \in \Bs$ with~$E^\Bs(c^\Bs,b)$ and~$(\As, \va, a) \bfe{k} (\Bs, \vb, b)$. Therefore by the induction hypothesis~$(\As,\va, a) \efe{k} (\Bs,\vb, b)$. Hence Duplicator can respond with~$b$ and win the remaining~$k$ rounds of the~$m + k + 1$-round game.

Secondly, we aim to show that~\ref{lab:efe} implies~\ref{lab:bfe}. Assume~$\Ava \efe{k + 1} \Bvb$, and there exists~$a \in \As$ with~$E^\As(c^\As,a)$. Take~$b$ be Duplicators response to a move of~$a$ by Spoiler. By definition~$E^\Bs(c^\Bs,b)$ and~$(\As,\va, a) \efe{k} (\Bs,\vb, b)$ and so by the induction hypothesis~$(\As,\va,a) \bfe{k} (\Bs,\vb,b)$.

So far, we have established the equivalence of~\ref{lab:efe} and~\ref{lab:bfe}. We now aim to show the equivalence of~\ref{lab:loge} and~\ref{lab:bfe}.

Firstly, to show~\ref{lab:loge} implies~\ref{lab:bfe}, assume~$\Ava$ agrees with~$\Bvb$ on bounded rank-$k + 1$ first order sentences. We must show~$\Ava \bfe{k + 1} \Bvb$. We explicitly check only the forth case, as the back case will be the same. For~$a \in \As$ with~$E^\As(c^\As,a)$, let~$\alpha$ be a (bounded) formula for the rank-$k$ type of~$a$. Then~$\Ava \models \exists x. E(c, x) \wedge \alpha(x)$, which has quantifier rank~$k + 1$. By assumption~$\Bvb \models \exists x. E(c,x) \wedge \alpha(x)$. Let~$b \in \Bs$ be the witness for the existential. The bounded rank-$k$ types of~$a$ and~$b$ are the same, therefore~$(\As, \va, a)$ and~$(\Bs, \vb, b)$ agree on quantifier rank~$k$ bounded first order sentences. By the induction hypothesis~$(\As, \va, a) \bfe{k} (\Bs, \vb, b)$.

Secondly, to show~\ref{lab:bfe} implies~\ref{lab:loge}, assume~$\Ava \bfe{k + 1} \Bvb$. Every rank-$k + 1$ bounded first order sentence is a Boolean combination of sentences of the form~$\exists x. E(c,x) \wedge \varphi$, so it suffices to restrict our attention to sentences of this form. Assume~$\Ava \models \exists x. E(c,x) \wedge \varphi(x)$. So there is an~$a$ such that~$\Ava \models \varphi(a)$ and~$E^\As(c^\As,a)$. By the forth condition, we have~$b \in \Bs$ such that~$E^\Bs(c^\Bs,b)$ and~$(\As, \va, a) \bfe{k} (\Bs, \vb, b)$. By the induction hypothesis~$(\As, \va,a)$ and~$(\Bs, \vb, b)$ agree on bounded first order sentences of depth~$k$. Hence~$\Bvb \models \varphi(b)$, and so~$\Bvb \models \exists x. E(c,x) \wedge \varphi(x)$. We can uses a symmetrical argument if we first assume~$\Bvb \models \exists x. E(c,x) \wedge \varphi(x)$, completing the proof.
\end{proof}

\subsection{The Bijection Game}
This section adapts the argument given in~\cite{Libkin2004} to the constraints of the bounded setting. We fix a bounded signature~$(\sg,\sgT)$.

\begin{definition}[The Bounded Quantifier Bijection Game]
For a $\Ava$ in $\CSmT$, and~$U \subseteq \As$, we say that $a \in \As$ is~$U$-\emph{accessible} if there exists~$E \in \sgT$ and~$t \in U$ such that~$E^\As(t,a)$. Let~$\Acc{\As}{U}$ denote the set of~$U$-accessible states in~$\As$.

The bijective Ehrenfeuct-\Fraisse game between~$\Ava$ and~$\Bvb$ in~$\CSmT$ proceeds as follows. A position of the game is a partial isomorphism between~$\Ava$ and~$\Bvb$. For round~$1 \leq i \leq m$ rounds of the game, Spoiler must play~$\va_i$ and Duplicator must respond with~$\vb_i$. In round~$i + 1 > m$ the game proceeds as follows:
\begin{enumerate}
    \item If the sets~$\Acc{\As}{\dom F_{i}}$ and~$\Acc{\Bs}{\cod F_{i}}$ have different cardinalities, Spoiler wins. Otherwise Duplicator picks a bijection:
    \begin{equation*}
        f : \Acc{\As}{\dom F_{i}} \rightarrow \Acc{\Bs}{\cod {F_i}}
    \end{equation*}
    \item Spoiler responds with $a \in \As$.
    \item We define $F_{i + 1} := F_i \cup \{ (a, f(a) \}$.
\end{enumerate}
The winning condition for Duplicator at the end of round~$i$ is that~$F_i$ is a partial isomorphism.
\end{definition}
We shall write~$\Ava \befe{k} \Bvb$ if Duplicator has a winning strategy in the~$m + k$-round bijective bounded Ehrenfeuct-\Fraisse game between~$\As$ and~$\Bs$.
\begin{definition}[Bounded Counting Logic]
First order logic with counting quantifiers is first order logic extended with quantifier:
\begin{equation*}
    \existc{i}{x} \varphi
\end{equation*}
with the informal reading ``there is at least~$i$ different elements~$x$ such that~$\varphi$ holds.

Bounded first order logic with counting quantifiers over a bounded signature~$(\sg, \sgT)$ is the fragment of first order logic with counting quantifiers where we restrict to quantification of the form:
\begin{equation*}
    \existc{i}{x} \beta \;\wedge\; \varphi
\end{equation*}
where~$\beta$ is of the form~$E(t,x)$ with~$E \in \sgT$ and~$t \neq x$.
\end{definition}

\begin{lemma}[General Bounding Formulae]
For~$k > 0$, formulae of the form:
\begin{equation*}
    \existc{i}{y} \left( \bigvee_{j \in \{ 1, \ldots, k \} }  \beta_j \right) \wedge \varphi
\end{equation*}
where~$\beta_j$ is of the form~$E(t,y)$ with~$E \in \sgT$ and~$t \neq y$,
are equivalent to formulae in the bounded fragment of the same quantifier depth.
\end{lemma}
\begin{proof}
The point is to avoid double counting. We proceed by induction on~$k$. In the base case, $k = 1$, and the original formula is already bounded.

For the inductive step, we can use the formula:
\begin{equation*}
    \bigvee_{m + n = l} \left[ \existc{m}{y} \left( \bigvee_{j \in \{ 1, \ldots, k\} } \beta_j \right) \wedge \neg \beta_{k + 1} \wedge \varphi \right] \wedge \left[ \existc{n}{y} \beta_{k + 1} \wedge \varphi \right]
\end{equation*}
and by the induction hypothesis, this is equivalent to a Boolean combination of bounded formulae of the same depth.
\end{proof}

\begin{definition}[Exact Counting Quantifiers]
For~$k > 0$, let
\begin{equation*}
    \existx{i}{y} \left( \bigvee_{j \in \{ 1, \ldots, k \} } \beta_j \right) \wedge \varphi
\end{equation*}
be the formula:
\begin{equation*}
    \left[\existc{i}{y} \left( \bigvee_{j in \{ 1, \ldots, k \} } \beta_j \right) \wedge \varphi \right] \wedge \neg \left[\existc{i}{y} \left( \bigvee_{j in \{ 1, \ldots, k \} } \beta_j \right) \wedge \varphi \right]
\end{equation*}
with intuitive reading ``there are exactly~$i$ elements~$x$ such that...''. Note this definition preserves bounded quantifier depth in the obvious sense.
\end{definition}

\begin{definition}[Accessibility Formulae]
For an $n$-tuple of variables~$\vx$, and~$y \not\in \supp{\vx}$, define the formula:
\begin{equation*}
    \acc{\vx}{y} := \bigvee_{E \in \sgT} \left( \bigvee_{1 \leq j \leq n} E(\vx_j,y) \vee \bigvee_{1 \leq j \leq m} E(c_j,y) \right)
\end{equation*}
Intuitively~$\acc{\vx}{y}$ defines the elements~$y$ that are accessible from~$\vx$ and the constant symbols in one step.
\end{definition}

\begin{definition}
\begin{enumerate}
    \item Note that as our signature is finite, there can be only finitely many atomic $m$-types, and each type can be described by a single quantifier free characteristic formula.
    Let~$\Scott{0}{m}{i}(\vx)$ be an enumeration of formulae in~$m$ variables defining the atomic type of~$\vx$. 
    \item Let~$\Scott{k + 1}{m}{i}(\vx)$ be an enumeration of the formulae in~$m$ variables of two forms:
    \begin{enumerate}
    \item For constants $p$, $l_p$ and~$i_p$ range over~$\mathbb{N}$, with~$p$ strictly positive:
    \begin{equation}
    \label{eq:Scott-inductive}
        \left( \bigwedge_{n = 1}^{p} \existx{l_n}{y} \acc{\vx}{y} \wedge \Scott{k}{m + 1}{i_n}(\vx,y) \right) \;\wedge\; \left( \forall y. \acc{\vx}{y} \rightarrow \bigvee_{n = 1}^{p} \Scott{k}{m + 1}{i_n}(\vx,y) \right)
    \end{equation}
    Intuitively, the first conjunct encodes the exact counts of accessible elements of certain~rank-$k$ $m + 1$-types, and the second conjunct ensures only elements of these types are accessible.
    \item For~$i$ ranging over~$\mathbb{N}$:
    \begin{equation}
    \label{eq:Scott-stuck}
        \left (\neg \exists y. \acc{\vx}{y} \right) \;\wedge\; \Scott{0}{m}{i}(\vx)
    \end{equation}
    Intuitively this formula describes the case where~$\vx$ cannot reach any states. The second conjunction pins down the atomic type of~$\vx$.
    \end{enumerate}
    The $\Scott{k}{m}{i}(\vx)$ are equivalent to a bounded formulae in counting logic by the previous lemmas.  
    Clearly, every tuple of elements in a finite structure can satisfy exactly one such condition.
\end{enumerate}

\end{definition}

\begin{definition}[Accessible Tuples and Elements]
For $\As$ in~$\CSmT$, we say that a tuple~$\vs$ is \emph{accessible} if for all~$i$:
\begin{equation*}
    \vs_i \in \Acc{\As}{\{ c_l^\As \mid 1 \leq l \leq m \} \cup \{ \vs_j \mid 1 \leq j < i \}}
\end{equation*}
Intuitively each element of the tuple is accessible from an ``earlier position''.

Given a structure~$\As$ and accessible tuple~$\vs$, we say that an element~$a$ is \emph{accessible} if the tuple~$\vs a$ is.
\end{definition}
The following technical lemma shows how satisfying a formula~$\Scott{k}{m}{i}$ can be unpacked to yield a strategy for playing the bijection game.
\begin{lemma}
\label{lem:key-step}
Let~$\Ava$ and~$\Bvb$ be finite structures in~$\CSmT$, $k,l \geq 0$, $\vs \in \As^l$ and~$\vt \in \Bs^l$. If~$\vs$ and~$\vt$ are accessible
tuples, and there exists~$\Scott{k}{l}{i}$ such that:
\begin{equation*}
    \Ava \models \Scott{k}{l}{i}(\vs) \;\mbox{ and }\; \Bvb \models \Scott{k}{l}{i}(\vt)
\end{equation*}
then the relation
\begin{equation*}
    \{ (c_j^\As, c_j^\Bs) \mid 1 \leq j \leq m \} \cup \{ (\vs_i,\vt_i) \mid 1 \leq i \leq l \}
\end{equation*}
is a winning position for Duplicator, and Duplicator has a winning strategy from that position for~$k$ additional rounds of the game.
\end{lemma}
\begin{proof}
We proceed by induction on~$k$. For the base case~$k = 0$, we note that
\begin{equation*}
    \{ (c_j^\As, c_j^\Bs) \mid 1 \leq j \leq m \}    
\end{equation*}
must be a bijection as both sides satisfy the same equalities between constants. It is a partial isomorphism as both sides agree on all atomic formulae.

For the inductive step, assume~$\Ava \models \Scott{k + 1}{l}{i}(\vs)$ and~$\Bvb \models \Scott{k + 1}{l}{i}(\vt)$. 
If $\Scott{k + 1}{l}{i}$ is of the form~\eqref{eq:Scott-stuck}, then the map~$\vs_j \mapsto \vt_j$ is a partial isomorphism, and Duplicator trivially wins the rest of the game as there are no available moves for Spoiler.

Otherwise, $\Scott{k + 1}{l}{i}$ is of the form~\eqref{eq:Scott-inductive}. 
Therefore there is a set of~$n$ indices~$J := \{ i_1,\ldots,i_n \}$ such that every accessible element~$a$ there is a unique~$j \in J$ such that $\Scott{k}{l + 1}{j}(\va,a)$. A similar argument holds for~$\Bs$. Furthermore, for~$j \in J$ the finite sets:
\begin{equation*}
    \{ a \mid a \mbox{ accessible}, \Ava \models \Scott{k}{l + 1}{j}(\vs,a) \} \;\mbox{ and }\; \{ b \mid b \mbox{ accessible}, \Bvb \models \Scott{k}{l + 1}{j}(\vt,b) \}
\end{equation*}
have the same cardinality. We can therefore choose a bijection~$f$ between the sets of accessible elements in the two structures, respecting the partitioning into~$n$ classes. If~$a$ is in the~$j^{th}$ class,
\begin{equation*}
    \Ava \models \Scott{k}{l + 1}{j}(\vs,a) \;\mbox{ and }\; \Bvb \models \Scott{k}{l + 1}{j}(\vt,f(a))
\end{equation*}
and so by the induction hypothesis, Duplicator has a winning strategy for~$k$ additional rounds of the game.
\end{proof}

\begin{theorem}
For finite~$\As$ and~$\Bs$ in~$\CSmT$ the following are equivalent:
\begin{enumerate}
    \item \label{en:counting-equiv} $\Ava$ and~$\Bvb$ agree on bounded first order sentences with counting quantifiers of depth at most~$k$.
    \item \label{en:bij-game-equiv} $\Ava \befe{k} \Bvb$.
\end{enumerate}
\end{theorem}
\begin{proof}
We first prove~\ref{en:bij-game-equiv} implies~\ref{en:counting-equiv} by induction. The base case is immediate from the definitions. For the inductive step, assume~\ref{en:bij-game-equiv} implies~\ref{en:counting-equiv} at~$k$, and aim to prove it holds for~$k + 1$. First consider a sentence of the form:
\begin{equation*}
    \existc{n}{x} \beta(x) \;\wedge\; \varphi(x)
\end{equation*}
where~$\beta$ is of the form~$E(c,x)$ for~$E \in \sgT$ and~$c$ a constant symbol. Let~$\vs$ be an~$n$-element tuple of distinct witnesses such that:
\begin{equation*}
    \Ava \models \beta(\vs_i) \;\mbox{ and }\; \Ava \models \varphi(\vs_i)
\end{equation*}
Therefore each~$\vs_i$ is in~$\Acc{\As}{\{ c \mid c \mbox{ a constant in } \sg \}}$, and so
as Duplicator has a winning strategies, there is a bijection~$f$ such that~$(\As,\va, a) \befe{k} (\Bs, \vb, f(a))$. Therefore by the induction hypothesis, $(\As,\va,a)$ and~$(\Bs,\vb, f(a))$ agree on bounded sentences of rank~$k$. Therefore:
\begin{equation*}
    \Bvb \models \beta(f(\vs_i)) \;\mbox{ and }\; \Bvb \models \varphi(f(\vs_i))
\end{equation*}
As~$f$ is a bijection, we therefore have:
\begin{equation*}
    \Bvb \models \existc{n}{x} \beta \;\wedge\; \varphi
\end{equation*}
The converse is proved similarly, using the inverse of the bijection from the winning strategy. Every bounded depth~$k + 1$ sentence is equivalent to a Boolean combination of such sentences, completing this direction of the proof.

For the other direction, we note that each structure satisfies a unique sentence~$\Scott{k}{0}{i}$. If~$\Ava$ and~$\Bvb$ agree on bounded first order sentences with counting quantifiers of depth at most~$k$ then there exists~$i$ such that:
\begin{equation*}
    \Ava \models \Scott{k}{0}{i} \;\mbox{ and }\; \Bvb \models \Scott{k}{0}{i}
\end{equation*}
Lemma~\ref{lem:key-step} then completes the proof.
\end{proof}

\end{document}